\documentclass[11pt, a4paper]{article}
 \usepackage{fixltx2e}
\usepackage{fix-cm}

\usepackage{amsmath, amsfonts, amssymb, mathrsfs}

 \usepackage{graphicx}
\graphicspath{{img/}}
 \usepackage{xcolor}
 \usepackage{xfrac}
 \usepackage{multi row}
  \usepackage{dsfont}
 \usepackage{arydshln} 
 \usepackage{natbib}
 \usepackage{enumerate}
 \usepackage[colorlinks,citecolor=blue,urlcolor=blue,linkcolor=blue,linktocpage=true]{hyperref}
\usepackage{authblk}
 	
 \usepackage{amsthm}

 \newcommand{\deriv}{\mathrm{d}} 
  
 \newcommand{\argmin}{\operatorname{arg\,min}~} 
 \newcommand{\simiid}{\overset{i.i.d.}{\sim}}
\newcommand{\Supp}{\operatorname{Supp}}
 \newcommand{\itemb}{\item[$\bullet$]}

\newcommand{\eps}{\varepsilon}
 
 \theoremstyle{plain}
 \newtheorem{thm}{Theorem}
 \newtheorem{proposition}[thm]{Proposition}
 \newtheorem{lemma}[thm]{Lemma}

\usepackage[margin=1.2in]{geometry}

 \title{Bayesian functional linear regression\\ with sparse step functions}
\author[1,2]{Paul-Marie Grollemund}
\author[2]{Christophe Abraham }
\author[2]{Me\"{\i}li Baragatti}
\author[3]{Pierre Pudlo}
\date{}

\affil[1]{IMAG UMR 5149, Universit\'e de Montpellier, CNRS,
  Place E. Bataillon, 34095 Montpellier CEDEX, France
  \href{mailto:paul-marie.grollemund@umontpellier.fr}{paul-marie.grollemund@umontpellier.fr}
}

\affil[2]{MISTEA UMR 729, INRA, Montpellier SupAgro,
  Place Pierre Viala, 34060 Montpellier CEDEX, France
    \href{mailto:christophe.abraham@supagro.fr}{christophe.abraham@supagro.fr}
    \href{mailto:meili.baragatti@supagro.fr}{meili.baragatti@supagro.fr}
}

\affil[3]{I2M UMR 7373, Aix-Marseille Universit\'e, CNRS, Centrale Marseille, 
  Rue F. Joliot Curie,
  13453 Marseille CEDEX 13, France
    \href{mailto:pierre.pudlo@univ-amu.fr}{pierre.pudlo@univ-amu.fr}
}

\begin{document}
\maketitle

\begin{abstract} 
    The functional linear regression model is a common tool to determine the relationship between a
  scalar outcome and a functional predictor seen as a function of time. This paper focuses on the
  Bayesian estimation of the support of the coefficient function. To
  this aim we propose a parsimonious and adaptive decomposition of the
  coefficient function as a step function, and a model including a
  prior distribution that we name Bayesian functional Linear
  regression with Sparse Step functions (Bliss). The aim of the method
  is to recover areas of time which influences the most the outcome. A
  Bayes estimator of the support is built with a specific loss
  function, as well as two Bayes estimators of the coefficient
  function, a first one which is smooth and a second one which is a
  step function. The performance of the proposed methodology is
  analysed on various synthetic datasets and is illustrated on a
  black P\'erigord truffle dataset to study the influence of rainfall
  on the production.\\


\end{abstract}

\section{Introduction} 

Consider that one wants to explain the final outcome $y$ of a
process along time (for instance the amount of some agricultural
production) thanks to what happened during the whole history (for instance,
the rainfall history, or temperature history). Among the statistical
learning methods, functional linear models \citep{Ramsay2005} aim at predicting a
scalar $y$ based on covariates $x_1(t), x_2(t),\ldots, x_q(t)$ lying
in a functional space, $L^2(\mathcal T)$ say, where $\mathcal T$ is an
interval of $\mathbb R$. If $x_{q+1},\ldots, x_p$ are additional
scalar covariates, the outcome $y$ is predicted linearly with
\begin{equation}
\widehat y =  \mu + \int_{\mathcal T}\beta_1(t) x_1(t) \deriv t +
\cdots
+ \int_{\mathcal T} \beta_q(t) x_q(t)\deriv t + \beta_{q+1}x_{q+1} +
\cdots + \beta_p x_p,
\label{eq:yhat}
\end{equation}
where $\mu$ is the intercept, $\beta_1(t),\ldots,\beta_q(t)$ the coefficient functions,
and $\beta_{q+1},\ldots,\beta_p$ the other (scalar) coefficients.  In this framework the
functional covariates $x_j(t)$ and the unknown coefficient functions $\beta_j(t)$ lie in
the $L^2(\mathcal T)$ functional space, thus we face a nonparametric problem. Standard
methods \citep{Ramsay2005} for estimating the $\beta_j(t)$'s, $1\le j\le q$, are based on
the expansion onto a given basis of $L^2(\mathcal T)$. See \citet{Reiss2015} for a
comprehensive scan of the methodology.
A question which arises naturally in many applied contexts is the
detection of periods of time which influence the most the final
outcome $y$. 
Note that each integral in \eqref{eq:yhat} is a weighted average of the
whole trajectory of $x_j(t)$, and does not identify any specific impact of
local period of the process.
These time periods might vary from one covariates to
another. For instance, in agricultural science, the final outcome may
depend on the amount of rainfall during a given period (e.g., to
prevent rotting), and the temperature during another (e.g., to prevent
freezing). Standard methods do not answer the above question, namely
to recover the support of the coefficient functions $\beta_j(t)$ with
the noticeable exception of \citet{picheny2016interpretable}.

Unlike the scalar-on-image models, we focus here on one-dimensional
functional covariates. When $\mathcal T$ is not a one dimensional
space, the problem becomes much more complex. The functional
covariates and the coefficient functions are all discretized,
e.g. via the pixels of the images, see \citet{Goldsmith2014, Li2015, kang2016scalar}. In these
two- or three-dimensional problems, because of the curse of
dimensionality, the points which are included in
the support of the coefficient functions follow a parametric
distribution, namely an Ising model. One important issue solved by these authors
is the sensitivity of the parameter estimate of the Ising model in the neighborhood of the phase
transition. 
%

When $\mathcal T$ is a one dimensional space, we can build
nonparametric estimates. In this vein, using the $L^1$-penalty to
achieve parsimony, the Flirti method of \citet{James2009} obtains an
estimate of the $\beta_j(t)$'s assuming they are sparse functions with
sparse derivatives. 
Nevertheless Flirti is difficult to calibrate: its numerical results
depend heavily on tuning parameters. From our experience,
Flirti's estimate is so sensitive to the values of the tuning
parameters that we can miss the range of good values with
cross-validation. The authors propose to rely on
cross-validation to set these tuning parameters. But, by definition, cross-validation assesses
the predictive performance of a model, see \citet{Arlot2010} and the
many references therein. And, of course, 
optimizing the performance regarding the prediction of $y$ does not
provide any guaranty regarding the support
estimate.
\citet{Zhou2013} propose a two-stage method to
estimate the coefficient function. Preliminarily, $\beta(t)$ is
expanded onto a B-spline basis to reduce the dimension of the
model. The first stage estimates the coefficients of the truncated
expansion onto the basis
using a lasso method to find the null intervals. Then, the second
stage refines the estimation of the null intervals and estimates the
magnitude of $\beta(t)$ for the rest of the support.  Another approach
to obtain parsimony is to rely on Fused lasso \citep{Tibshirani2005}:
if we discretize the covariate functions and the coefficient
function as described in \citet{James2009}, the penalization of Fused
lasso induces parsimony in the coefficients. But, once again the calibration
of the penalization is performed using cross-validation which targets
predictive perfomance rather than the accuracy of the support estimate.

In this paper, we propose Bayesian estimates of both the supports and
the coefficient functions $\beta_i(t)$'s. To keep the dimension of the
parameter as low as possible, we stay with the simplest and the most
parsimonious shape of the coefficient function over its
support. Hence, conditionally on the support, the coefficient functions
$\beta_j(t)$'s are supposed to be step function (piecewise constant
function can be described with a minimal number of parameters). We can
decompose any step function $\beta(t)$ as
\[
\beta(t) = \sum_{k=1}^K \beta^\ast_k \frac{1}{|\mathcal{I}_k|}\mathbf
1\{t\in \mathcal I_k\}
\]
where $\mathcal I_1,\ldots, \mathcal I_K$ are intervals of $\mathcal T$, $|\mathcal I_k|$ is the
length of the interval $\mathcal I_k$ and $\beta^\ast_k$ are the coefficients
of the expansion. The support is the union of all $\mathcal I_k$'s if
the coefficients $\beta^\ast_k$ are non null. Period of
times which does not influence the outcome are outside the support. 
The above model has another advantage: such step functions change
values abruptly from $0$ to a non null value. Hence their supports are
relatively clear. On the contrary, if we have at our disposal a smooth
estimate of a coefficient function $\beta_j(t)$ in the model given by
\eqref{eq:yhat}, the support of the estimate is the whole $\mathcal T$
and we have to find regions where the estimate is not
significantly different from $0$.
Moreover, with a full Bayesian procedure, we can evaluate the
uncertainty on the estimates of the support and the values of the
coefficient functions.

The paper is organized as follows. Section~\ref{sec:method} presents the Bayesian modelling,
including the prior distribution in \ref{sec:model}, the support estimate in \ref{sec:support}
and the coefficient function estimate in \ref{sec:beta}. Section~\ref{sec:simulation} is
devoted to the study of numerical results on synthetic data, with comparison to other methods
and sensibility to the tuning of the hyperparameters of the prior. Section~\ref{sec:reel}
details the results of Bliss on a dataset concerning the influence of rainfall on the growth of
the black P\'erigord truffle.  

\section{The Bliss method}
\label{sec:method}

We present the hierarchical Bayesian model in Section~\ref{sec:model},
the Bayes estimate of the support in Section~\ref{sec:support} and two
Bayes estimates of the coefficient function in
Section~\ref{sec:beta}. The implementation and visualization details
are given at the end of this second part.

\subsection{Reducing the model}
\label{sec:reducing}

Assume we have observed $n$ independent replicates $y_i$ ($1\le i\le
n$) of the outcome, explained with the functional covariates $x_{ij}(t)$
($1\le i \le n$, $1\le j \le q$) and the scalar covariates $x_{ij}$
($1\le i \le n$, $q+1\le j \le p$). 
The whole dataset will be denoted  $\mathcal D$ in what follows. Let
us denote by $x_i=\{x_{i1}(t),\ldots, x_{iq}(t), x_{i,q+1}, x_{ip}\}$
the set of all covariates, and by $\theta$ the set of all parameters,
namely $\{ \beta_1(t),\ldots, \beta_q(t), \beta_{q+1},\ldots, \beta_p, \mu,
\sigma^2\}$, where $\sigma^2$ is a variance parameter.
We resort to the Gaussian likelihood defined as
\begin{equation}
  \label{eq:model_bayes_p}
  y_i | x_i,\theta  \overset{\text{ind}}{\sim} 
  \mathcal{N} \left( \mu + \sum_{j=1}^q\int_{\mathcal T}\beta_j(t)
    x_{ij}(t) \deriv t +
\sum_{j=q+1}^p\beta_{j}x_{ij},~\sigma^2 \right), \qquad i = 1,\dots,n.
\end{equation}
If we set a prior on the parameter $\theta$ which includes all
$\beta_j(t)$'s, $\beta_j$, $\mu$ and $\sigma^2$, we can recover the full
posterior from the following conditional distributions (both theoretically
and practically with a Gibbs sampler) :
\begin{align*}
  \beta_j(t), \mu, \sigma^2 &\,|\, \mathcal D, \beta_{-j}
  \\
  \beta_j, \mu, \sigma^2 &\,|\, \mathcal D, \beta_{-j}
\end{align*}
where $\beta_{-j}$ represents the set of $\beta$-parameters except $\beta_j$
or $\beta_j(t)$. 
Hence we can reduce the problem to a single functional covariate and
no scalar covariate.
The model we have to study becomes 
\begin{equation}
  \label{eq:model_bayes_1}
  y_i | x_i(t),\mu, \beta(t), \sigma^2  \overset{\text{ind}}{\sim} 
  \mathcal{N} \left( \mu + \int_{\mathcal{T}} \beta(t) x_i(t)  \deriv
    t \, ,~\sigma^2 \right), \qquad i = 1,\dots,n,
\end{equation}
with a single functional covariate $x_i(t)$.

\subsection{Model on a single functional covariate}
\label{sec:model}
For parsimony we seek the coefficient function $\beta(t)$ in the following set of sparse step functions 
\begin{equation}
  \label{eq:calBK}
  \mathcal E_K=\left\{ \sum_{k=1}^K \beta^\ast_k \frac{1}{|\mathcal
      I_k|}\mathds{1}\left\{t \in \mathcal I_k \right\}: \ 
  \mathcal I_1,\ldots, \mathcal I_K\text{ intervals }\subset\mathcal T, \beta^\ast_1,\ldots,\beta^\ast_K\in\mathbb R\right\}
\end{equation}
where $K$ is a hyperparameter that counts the number of intervals
required to define the function. Note that we do not make any
assumptions regarding the intervals
$\mathcal I_1,\ldots,\mathcal I_K$. First they do not form a partition
of $\mathcal T$. As a consequence, a function $\beta(t)$ in
$\mathcal E_K$ is piecewise constant and null outside the union of the
intervals $\mathcal I_k$, $k=1,\ldots,K$.  This union is the support
of $\beta(t)$, hence the model includes an explicit description of the
support. Second the intervals $\mathcal I_1, \ldots, \mathcal I_K$ can
even overlap to ease the parametrization of the intervals: we do not
have to add constraints on the parametrization to remove possible
overlaps.

Now if we pick a function $\beta(t)\in \mathcal E_K$ with
\begin{equation}
  \label{eq:beta_forme}
  \beta(t) = \sum_{k=1}^K \beta^*_k \frac{1}{|\mathcal I_k|}\mathds{1}\left\{t \in \mathcal I_k \right\},
\end{equation}
the integral of the covariate functions $x_i(t)$ against $\beta(t)$ becomes a linear combination of partial integrals of the covariate function over the intervals $\mathcal I_k$ and we predict $y_i$ with
\[
\widehat{y_i} = \mu + \sum_{k=1}^K \beta^\ast_k\, x_i(\mathcal I_k),\quad
\text{where } 
x_i(\mathcal I_k)=\frac{1}{|\mathcal I_k|}\int_{\mathcal I_k} x_i(t)  \deriv t.
\]
Thus, given the intervals $\mathcal I_1,\ldots, \mathcal I_K$, we face a multivariate linear model with the usual Gaussian likelihood. 

It remains to set a parametrization on $\mathcal E_K$ and a prior distribution. Each interval $\mathcal I_k$ is parametrized with its center $m_k$ and its half length $\ell_k$:
\begin{equation}
  \label{eq:intervalle}
  \mathcal I_k = \left[ m_k - \ell_k, m_k + \ell_k \right] \cap \mathcal T.
\end{equation}
As a result, when $K$ is fixed, the parameter of the model is
\[\theta=(m_1,\ldots, m_K, \ell_1,\ldots, \ell_K, \beta^\ast_1,\ldots,
\beta^\ast_K, \mu, 
\sigma^2).\]

We first define the prior on the support, that is to say on the
intervals $\mathcal I_k$. The prior on the center of each interval is uniformly
distributed on the whole range of time $\mathcal T$. This uniform
prior does not promote any particular region of $\mathcal
T$. Furthermore, the prior on the half-length of the interval $\mathcal I_k$ is
the Gamma distribution $\Gamma(a, b)$. To understand this prior and set hyperparameters $a$ and $b$, we introduce the prior probability that a given $t\in\mathcal T$ is in the support, namely
\begin{equation}
\alpha(t) = \int_{\Theta_K}\mathbf 1\{t \in S_\theta\} \pi_K(\theta) \deriv \theta
\label{eq:alpha}
\end{equation}
where $\pi_K$ is the prior distribution on the range of parameters $\Theta_K$ of dimension $3K+2$, and where $S_\theta=\Supp(\beta_\theta)$ is the support of $\beta_\theta(t)$ that is to say the union of the $\mathcal I_k$'s. 
The value of $\alpha(t)$ depends on hyperparameters $a$ and
$b$. These parameters should be fixed with the help of prior knowledge
on $\alpha(t)$.

\begin{figure}[t!]
  \centering
  \includegraphics[height=.8\textwidth]{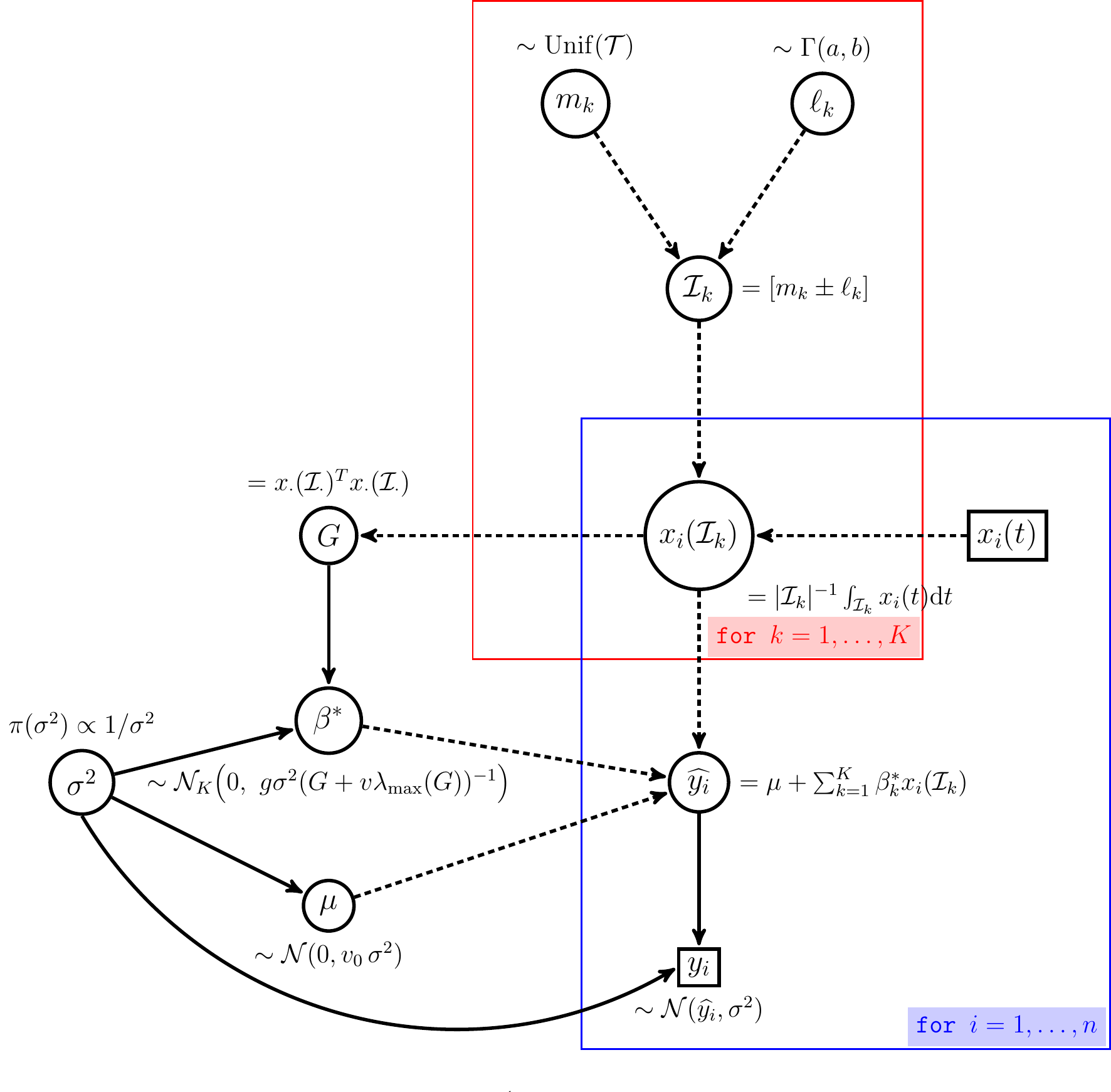} \\
  \caption{{The full Bayesian model.} \it The coefficient function
    $\beta(t)=\sum_{k=1}^K\beta^\ast_k \mathds{1}\{t\in \mathcal
    I_k\}/|\mathcal I_k|$ defines both a projection the covariate
    functions $x_i(t)$ onto $\mathbb R^K$ by averaging the function
    over each interval $\mathcal I_k$ and a prediction $\widehat{y_i}$
    which depends on the vector $\beta^\ast=(\beta^\ast_1,\ldots,
    \beta^\ast_K)$ and the intercept $\mu$. } 
  \label{fig:model}
\end{figure}

Given the intervals, or equivalently, given the $m_k$'s and
$\ell_k$'s, the functional linear model becomes a multivariate linear model with
$x_i(\mathcal I_k)$ as scalar covariates. We could have set a standard and well-understood prior on
$\beta^\ast|(\mathcal I_k)_{1\le k \le K}$, namely the $g-$Zellner prior, with $g=n$ in order to
define a vaguely informative prior.
More specifically, the design matrix given the intervals is 
\[
x_\cdot(\mathcal I_\cdot) = \left\{ x_i(\mathcal I_k),\ 1\le i\le n,\ 1\le k\le K\right\}.
\]
And the $g$-Zellner prior, with $g=n$ is given by
\begin{equation} \label{eq:Zellner}
  \pi(\sigma^2)  \propto 1/\sigma^2,
  \quad 
  \beta^\ast | \sigma^2 \sim \mathcal N_K\bigg(0, n\sigma^2 G^{-1}\bigg)
\end{equation}
where $\beta^* = (\beta_1^*, \dots, \beta_K^*)$ and
$G =x_\cdot(\mathcal I_\cdot)^T x_\cdot(\mathcal I_\cdot)$ is the Gram
matrix.  But, depending on the intervals $\mathcal I_k$, the
covariates $x_i(\mathcal I_k)$ can be highly 
correlated. (We recall here that the functional covariate can have
autocorrelation and that the intervals can overlap.)  That is why, in
this setting, the Gram matrix
$G= x_\cdot(\mathcal I_\cdot)^T x_\cdot(\mathcal I_\cdot)$ can be
ill-conditioned, that is to say not numerically invertible and we
cannot resort to the $g-$Zellner prior. To solve this issue we have to
decrease the condition number of $G$, and apply a Tikhonov
regularization. The resulting prior is a ridge-Zellner prior
\citep{Baragatti2012} replaces $G$ by $G+\eta I$ in
\eqref{eq:Zellner}, where $\eta$ is some scalar tuning the amount of
regularization and $I$ is the identity matrix.  Adding the $\eta I$
matrix shifts all eigenvalues of the Gram matrix by $\eta$. In order
to obtain a well-conditioned matrix, we decided to fix $\eta$ with the
help of the largest eigenvalue of the Gram matrix,
$\lambda_\text{max}(G)$ and
to set
$\eta=v \lambda_\text{max}(G)$ where $v$ is an hyperparameter of the model.

To sum up the above, the prior distribution on $\Theta_K$ is
\begin{align}
  \mu | \sigma^2      & \sim \mathcal{N}\left(0, v_0 \sigma^2\right), \notag \\
  \beta^* | \sigma^2 & \sim \mathcal{N}_K \left(0, n\sigma^2  (G + v\lambda_\text{max}(G)
                       I )^{-1}  \right), 
   \text{ where } G= x_\cdot(\mathcal I_\cdot)^T
x_\cdot(\mathcal I_\cdot), \notag \\
	\pi(\sigma^2)   & \propto 1/\sigma^2, 
 \label{eq:model_complet} \\
  m_k  &\simiid \text{Unif} \left( \mathcal{T} \right), \quad  k=1,\ldots, K, \notag \\
  \ell_k & \simiid \Gamma(a, 1), \quad  k=1,\ldots, K, \notag
\end{align}
The resulting Bayesian modelling is given in Figure~\ref{fig:model} and depends on
hyperparameters which are $v_0, v, a$ and $K$.  We denote by $\pi_K(\theta)$ and
$\pi_K(\theta|\mathcal D)$ the prior and the posterior distributions. 
We propose below default values for the hyperparameters $v_0, v, a$,
see Section \ref{choixhyperparametres} for numerical results that
supports this proposal.
\begin{itemize}
\item The parameter $v_0$ drives the prior information we put on the intercept
  $\mu$. This is clearly not the most important hyperparameter since
  we expect important information regarding $\mu$ in the
  likelihood. We recommend using $v_0=100\times \bar{y}^2$, where
  $\bar{y}$ is the average of the outcome on the dataset. Even if it may look like we set the prior with the
  current data, the resulting prior is vaguely non-informative.
\item The parameter $v$ is more difficult to set: it tunes the amount
  of regularization in the $g$-Zellner prior. Our set of numerical
  studies indicates, see Section~\ref{sec:simulation} below, that
  $v=5$ is a good value.
\item The parameter $a$ sets the prior length of an interval of the
  support. It should depend on the number $K$ of intervals. We
  recommend the value $a=(5K)^{-1}$ so that the average length of an
  interval from the prior distribution is proportional to $1/K$. Our
  numerical studies that constant $5$ in the above recommandation does
  not drastically influence the results.
\end{itemize}

Finally, the hyperparameter $K$ drives the number of intervals, thus
the dimension of $\Theta_K$. We can put an extra prior distribution on
$K$ and perform Bayesian model choice either to infer $K$ or to
aggregate posteriors coming from various values of $K$. There is a ban
on the use of improper prior together with Bayesian model choice (or
Bayes factor) because of the Jeffrey-Lindley paradox \citep[see,
e.g.][Section 5.2]{robert2007bayesian}. And a careful reader would
notice here the improper prior on $\sigma^2$. But it does not prohibit the
use of Bayesian choice because it is a parameter common to all models
(i.e., to all values of $K$ here).

\subsection{Estimation of the support}
\label{sec:support}
Regarding the inference of the support, an interesting quantity is the
posterior probability that a given $t\in\mathcal T$ is in the
support. It can be defined as the prior probability in
\eqref{eq:alpha}, that is to say
\begin{equation}
\alpha(t|\mathcal D) = \int_{\Theta_K}\mathbf 1\{t \in S_\theta\} \pi_K(\theta|\mathcal D) \deriv \theta.
\label{eq:alpha.posterior}
\end{equation}
Both functions $\alpha(t)$ and $\alpha(t|\mathcal D)$ can be easily
computed with a sample from the prior and the posterior
respectively. They are also relatively easy to interpret in term of
marginal distribution of the support: fix $t\in\mathcal T$,
\begin{itemize}
\item $\alpha(t)$ is the prior probability that $t$ is in the support of the coefficient function and
\item $\alpha(t|\mathcal D)$ is the posterior probability of the same event.
\end{itemize}

Now  let $L_\gamma(S, S_\theta)$ be the loss function given by 
\begin{equation}\label{eq:loss.support}
L_\gamma(S, S_\theta) = \gamma\int_0^1 \mathbf 1\{ t\in S\setminus S_\theta\} \deriv t + 
                  (1-\gamma)\int_0^1 \mathbf 1\{ t\in S_\theta\setminus S\} \deriv t 
\end{equation}
where $S_\theta= \Supp(\beta_\theta)$ is the  support of $\beta_\theta(t)$, the coefficient function as parametrized in \eqref{eq:beta_forme} and where $\gamma$ is a tuning parameter in $[0;1]$.
Actually, there is two type of errors when estimating the support:
\begin{itemize}
\item error of type I: a point $t\in\mathcal T$ which is really in the support $S_\theta$ has not been included in the estimate,
\item error of type II: a point $t\in\mathcal T$ has been included in the support estimate but does not lie into the real support $S_\theta$
\end{itemize}
and the tuning parameter $\gamma$ allows to set different weights on both types of error. Note that, when $\gamma=1/2$, the loss function is one half of the Lebesgue measure of the symmetric difference $S\Delta S_\theta$.

Bayes estimates are obtained by minimizing a loss function integrated
with respect to the posterior distribution, see
\citet{robert2007bayesian}. Hence, in this situation, Bayes estimates of the support are given by
\begin{equation}\label{eq:optim.support}
\widehat{S}_\gamma(\mathcal D) \in \underset{ S \subset \mathcal T }{\argmin} 
\int_{\Theta_K} L_\gamma(S, S_\theta) \pi_K(\theta|\mathcal D)\deriv \theta.
\end{equation}
The following theorem shows the existence of the Bayes estimate and how to compute it from $\alpha(t|\mathcal D)$. 
\begin{thm} \label{thm:support}
  The level set of $\alpha(t|\mathcal D)$ defined by 
  \[ \widehat{S}_\gamma(\mathcal D) = \{t\in\mathcal T\ :\ \alpha(t|\mathcal D)\ge \gamma \} \]
  is a Bayes estimate associated to the above loss $L_\gamma(S, S_\theta)$.
  Moreover, up to a set of null Lebesgue measure, any Bayes estimate
  $\widehat{S}_\gamma(\mathcal D)$ that solves the optimisation
  problem given in \eqref{eq:optim.support} satisfies
  \[
  \{t\in\mathcal T\ :\ \alpha(t|\mathcal D)> \gamma\}
  \subset\widehat{S}_\gamma(\mathcal D)  \subset \{t\in\mathcal T\ :\ \alpha(t|\mathcal D)\ge \gamma\}.
  \]
\end{thm}
The proof of the above theorem is given in Appendix \ref{sec:thm1}. Although simple-looking, the proof requires some caution because sets should be Borelian sets. Note that, when we try to avoid completely errors of type I (resp. type II) by setting $\gamma=0$ (resp. $\gamma=1$), the support estimate is $\mathcal T$ (resp. $\emptyset$). 
Additionally Theorem~\ref{thm:support} shows how we should interpret the posterior probability $\alpha(t|\mathcal D)$ and that its plot may be one important output of the Bayesian analysis proposed in this paper: it measures the evidence that a given point is in the support of the coefficient function. 
Finally, note that the number of intervals in the support estimate $\widehat{S}_\gamma(\mathcal
D)$ can, and is often different from the value of $K$ (because intervals can overlap).

\subsection{Estimation of the coefficient function}\label{sec:beta}
The Bayesian modelling given in Section~\ref{sec:model} was mainly designed to estimate
the support of the coefficient function. We can nevertheless provide Bayes estimates of
the coefficient function. We propose here two Bayes estimates of the coefficient
function. The first one, given in Equation~\eqref{eq:BayesL2} is a smooth estimate,
whereas the second estimate, given in Proposition~\ref{pro:3}, is a stepwise estimate
which is parsimonious and may be more easily interpreted.

With the default quadratic loss, a Bayes estimate is defined as
\begin{equation}
\widehat{\beta}_{L^2}(\cdot) \in \underset{d(\cdot) \in L^2(\mathcal T)}{\operatorname{arg\,min}} \iint \left(\beta_\theta(t) - d(t)\right)^2 \deriv t\ \pi_K(\theta|\mathcal D)\deriv \theta
\label{eq:BayesL2}
\end{equation}
where $\beta_\theta(t)$ is the coefficient function as parametrized in \eqref{eq:beta_forme}. 
At least heuristically $\widehat{\beta}_{L^2}(\cdot)$ is the average of $\beta_\theta(\cdot)$ over the posterior distribution $\pi_K(\theta|\mathcal D)$, though the average of functions taking values in $L^2(\mathcal T)$ under some probability distribution is hard to define (using either Bochner or Pettis integrals).
In this simple setting we can claim the following (see Appendix \ref{sec:proof1} for the proof). 
\begin{proposition}\label{pro:1}
  Let $\| \cdot \|$ be the norm of $L^2(\mathcal T)$. If
  $
  \int \left\|\beta_\theta(\cdot)\right\| \pi_K(\theta|\mathcal D)\deriv\theta < \infty,
  $
  then the estimate defined by 
  \begin{equation}\label{eq:betaL2}
    \widehat{\beta}_{L^2}(t) = \int \beta_\theta(t) \pi_K(\theta|\mathcal D)\deriv \theta, \quad\ t \in \mathcal T,
  \end{equation}
  is in $L^2(\mathcal T)$ and solves the optimization problem~\eqref{eq:BayesL2}.
\end{proposition}
Averages such as \eqref{eq:betaL2} belong to the closure of the convex hull of the support
$\mathcal E_K$ of the posterior distribution. We can prove (see Proposition~\ref{pro:topo}
in Appendix~\ref{sec:proof3}) that the convex hull of $\mathcal E_K$ is the set
$\mathcal E = \cup_{K=1}^\infty \mathcal E_K$ of step functions on $\mathcal T$, and the
closure of $\mathcal E$ is $L^2(\mathcal T)$. Hence the only guarantee we have on
$\widehat{\beta}_{L^2}$ as defined in \eqref{eq:betaL2} is that $\widehat{\beta}_{L^2}$
lies in $L^2(\mathcal T)$, a much larger space than the set of step functions.  Though not
shown here, integrating the $\beta_\theta(t)$'s over $\theta$ with respect to the
posterior distribution has regularizing properties, and the Bayes estimate
$\widehat{\beta}_{L^2}(t)$ is smooth.


To obtain an estimate lying in the set of step functions, namely $\mathcal E$, we can consider
the projection of  $\widehat{\beta}_{L^2}$ onto the set $\mathcal E_{K_0}$ for a suitable
value of $K_0$ possibly different to $K$. However, due to the topological properties of $L^2(\mathcal T)$ and 
$\mathcal E_{K_0}$, the projection of
$\widehat{\beta}_{L^2}$ onto the set $\mathcal E_{K_0}$ does not always exist (see Appendix~\ref{sec:proof3}). To address
this problem, we introduce a subset $\mathcal E_{K_0}^{\eps}$ of $\mathcal
E_{K_0}$, where $\varepsilon>0$ is a tuning parameter. 
Let $\mathcal F^{\eps}$ denote the set of step functions $\beta(t)\in L^2(\mathcal T)$ which can be
written as
\[
\beta(t) = \sum \beta^\dag_k \mathds 1\{t\in J_k\}
\] 
where the intervals $J_k$'s are mutually disjoint and each of the lengths are greater than $\eps$.
The set $\mathcal E_{K_0}^{\eps}$ is now defined as $\mathcal F^{\eps} \cap \mathcal
E_{K_0}$. By considering this set, we remove from $\mathcal{E}_K$ the step functions which
have intervals of very small length, and we can prove the following. 
\begin{proposition} \label{pro:3}
  Let $K_0 \geq 1$ and $\varepsilon>0$. 
  \begin{itemize}
  \item[\it (i)] The function $d(\cdot)\mapsto
    \|d(\cdot)-\widehat{\beta}_{L^2}(\cdot)\|^2$ admits a minimum on $\mathcal
    E_{K_0}^{\eps}$. Thus a projection of $\widehat{\beta}_{L^2}(\cdot)$ onto this
    set,  defined by 
    \begin{equation}\label{eq:Bayes3}
      \widehat{\beta}_{K_0}^{\eps}(\cdot) \in 
      \underset{d(\cdot)\in\mathcal E_{K_0}^{\eps}}{\operatorname{arg\,min}}
      \ \|d(\cdot)-\widehat{\beta}_{L^2}(\cdot)\|^2, 
\end{equation} always exists.
  \item[\it (ii)]  The estimate $\widehat{\beta}_{K_0}^{\eps}(\cdot)$ is a true Bayes estimate with loss function
  \begin{equation}
  \label{eq:cout}
  L_{K_0}^{\eps} \big( d(\cdot), \beta(\cdot) \big)= \begin{cases}
    \left\| d(\cdot) - \beta(\cdot) \right\|^2= \int_{\mathcal T} \left(\beta(t)-d(t)\right)^2\deriv t
    & \text{if } \beta\in \mathcal E_{K_0}^{\eps},
    \\
    +\infty & \text{otherwise}.
    \end{cases}
\end{equation}
That is to say
\[
\widehat{\beta}_{K_0}^{\eps}(\cdot) \in \underset{d(\cdot)\in L^2(\mathcal T)}{\operatorname{arg\,min}}\ 
\int L_{K_0}^{\eps} \big( d(\cdot), \beta_\theta(\cdot) \big) \pi_K(\theta|\mathcal D)\deriv\theta.
\]
  \end{itemize}
\end{proposition}
Finally one should note that the support of the Bliss estimate given
in Proposition~\ref{pro:3} provides another estimate of the support,
which differs from the Bayes estimate introduced in
Section~\ref{sec:support}. Obviously, real Bayes estimates, which
optimizes the loss integrated over the posterior distribution, are by
construction better estimates. Another possible alternative would be
the definition of an estimate of the coefficient function whose
support is given by one of the Bayes estimates defined in
Theorem~\ref{thm:support}. But such estimates do not account for the
inferential error regarding the support. Hence we believed that, when
it comes to estimating the coefficient function, the Bayes estimates
proposed in this Section are better than other candidates and achieve
a trade off between inferential errors on its support and prediction
accuracy on new data.
\subsection{Implementation}
\label{sec:implementation}
The full posterior distribution can be written explicitly from the Bayesian model given in Equations~\eqref{eq:model_complet}.
As usual with hierarchical models, sampling from the posterior distribution $\pi_K(\theta|\mathcal D)$ can be done with a Gibbs algorithm \citep[see, e.g.,][Chapter 7]{robert2013monte}. The details of the MCMC algorithm are given in Appendix~\ref{app:conditional}.

Now let $\theta(s)$, $s=1, \ldots, N$, denote the output of the MCMC sampler after the
burn-in period.  The computation of the Bayes estimate $\widehat{S}_\gamma(\mathcal D)$ of
the support as defined in Theorem~\ref{thm:support} depends on the probabilities
$\alpha(t|\mathcal{D})$. With the Monte Carlo sample from the MCMC, we can easily
approximate these posterior probabilities by the frequencies
\[
	{\alpha}(t |\mathcal{D}) \approx \frac{1}{N} \sum_{s=1}^N\mathds{1}\{\beta_{\theta(s)}(t) \neq 0 \}. 
\]

What remains to be computed are the approximations of $\widehat{\beta}_{L^2}(\cdot)$ and
$\widehat{\beta}_{K_0}^{\eps}(\cdot)$ based on the MCMC sample. First, the Monte
Carlo approximation of \eqref{eq:betaL2} is given by 
\[
\widehat{\beta}_{L^2}(t) \approx \frac1N \sum_{s=1}^N \beta_{\theta(s)}(t).
\]
And the more interesting Bayes estimate $\widehat{\beta}_{K_0}^{\eps}(\cdot)$ can be computed by
minimizing 
\[
\left\| d(\cdot) - \frac1N \sum_{s=1}^N \beta_{\theta(s)}(\cdot)  \right\|^2
\]
over the set $\mathcal E_{K_0}^{\eps}$. To this end we run a Simulated annealing
algorithm \citep{Kirkpatrick1983}, described in Appendix \ref{ann:SANN}.

We also provide a striking graphical display of the posterior distribution on the set
$\mathcal E_K$ with a heat map. More precisely, the aim is to sketch all marginal
posterior distributions $\pi^t_K(\cdot|\mathcal D)$ of $\beta_\theta(t)$ for any value of
$t\in\mathcal T$ in one single figure. To this end we introduce the probability measure
$Q$ on $\mathcal T\times\mathbb R$ defined as follows. Its marginal distribution over
$\mathcal T$ is uniform, and given the value $t$ of the first coordinate, the second
coordinate is distributed according to the posterior distribution of $\beta(t)$. In other
words,
\[
	(t,b)\sim Q \quad \iff \quad t\sim \text{Unif}(\mathcal T),\ b|t\sim
    \pi^t_K(\cdot|\mathcal D).
\] 
We can easily derive an empirical approximation of $Q$ from the MCMC sample $\{\theta(s)\}$ of the posterior. Indeed, the first marginal distribution of $Q$, namely $\text{Unif}(\mathcal T)$ can be approximated by a regular grid $t_i$, $i=1,\ldots,M$. And, for each value of $i$, set $b_{is}=\beta_{\theta(s)}(t_i)$, $s=1,\ldots, N$. The resulting empirical measure is
\[
	\widehat{Q} = \frac{1}{M\,N}\sum_{i=1,\ldots,M}\sum_{j=1,\ldots,N}\delta_{(t_i, b_{is})},
\]
where $\delta_{(t,b)}$ is the Dirac measure at $(t,b)$. The graphical
display we propose is representing $\widehat{Q}$ with a heat map on
$\mathcal T \times \mathbb R$. Each small area of $\mathcal T \times
\mathbb R$ is thus colored according to its $\widehat
Q$-probability. This should be done cautiously as the marginal
posterior distribution $\pi^t_K(\cdot|\mathcal D)$ has a point mass at
zero: $\pi^t_K(b=0|\mathcal D)>0$ by construction of the prior
distribution. Finally the color scale can be any monotone function of the
probabilities, in particular non linear functions to handle the atom
at $0$. Examples are provided in Section~\ref{sec:simulation}
in Figures~\ref{fig:estim_coeff} and \ref{fig:estim_coeff2}.

\section{Simulation study}
\label{sec:simulation}
In this section, the performance of Bliss is evaluated and compared to three competitors:
FDA \citep{Ramsay2005}, Fused lasso \citep{Tibshirani2005} and Flirti \citep{James2009},
using simulated datasets. 
\begin{figure}	
	\centering
	\includegraphics[width=7cm]{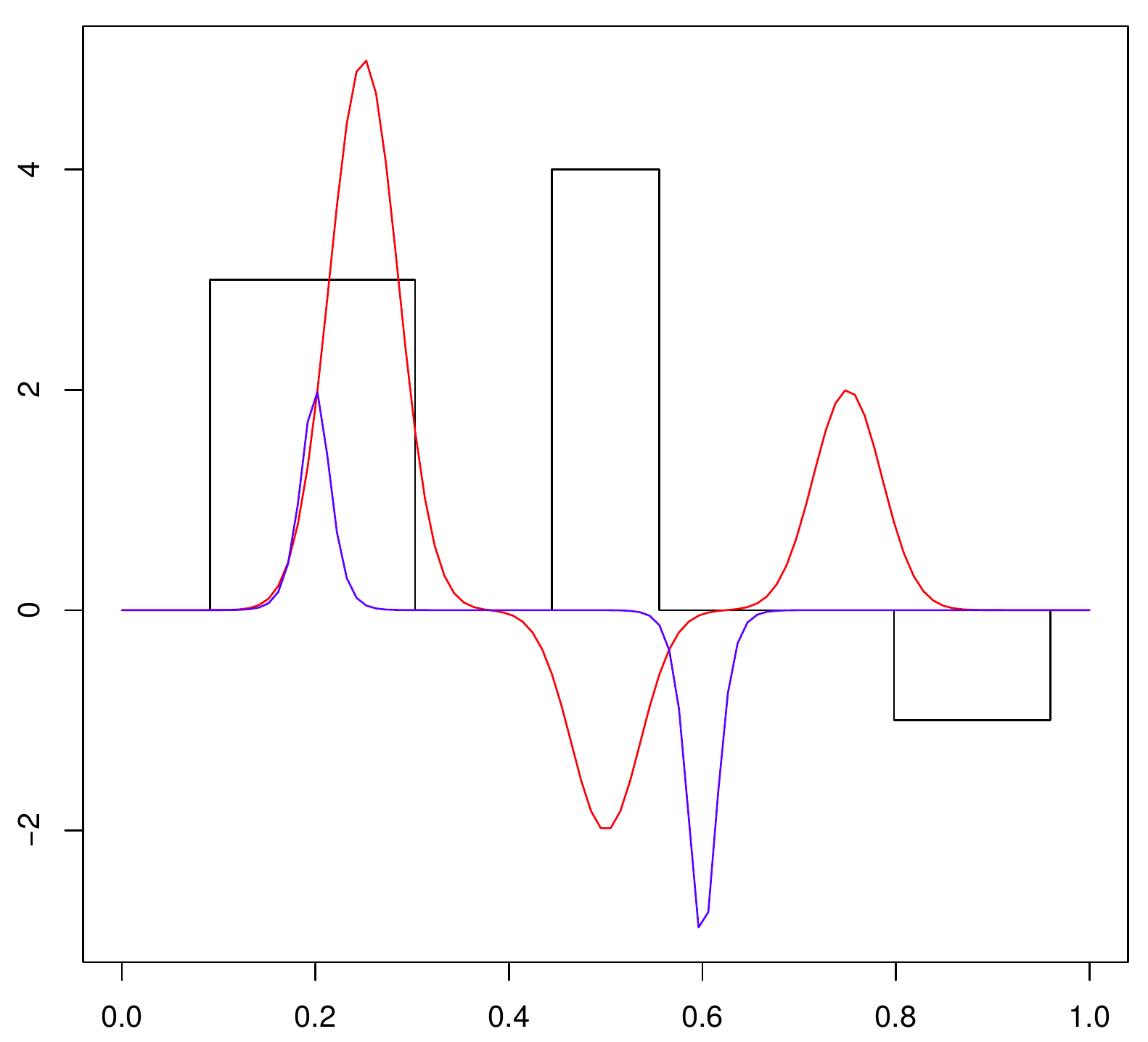} \\
	\caption{{Coefficient functions for numerical illustrations.} \it The black (resp. red and blue)
          curve corresponds to the coefficient function of Shape~1 (resp. 2 and 3).}
	\label{fig:fonction_coef}
\end{figure}

\begin{figure}[t]
  \centering
  \begin{tabular}{cc}
    \includegraphics[width=.45\textwidth]{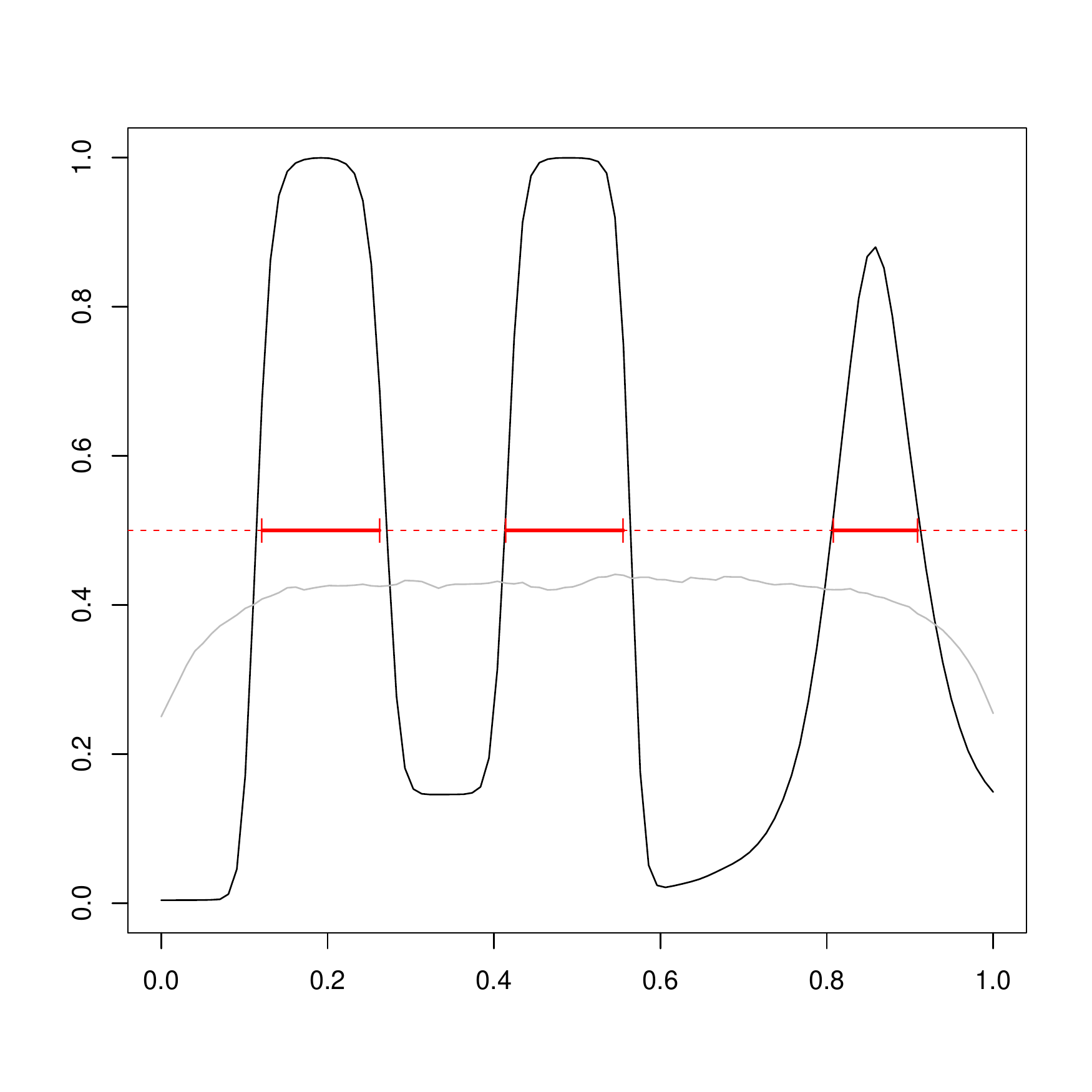}
    &\includegraphics[width=.45\textwidth]{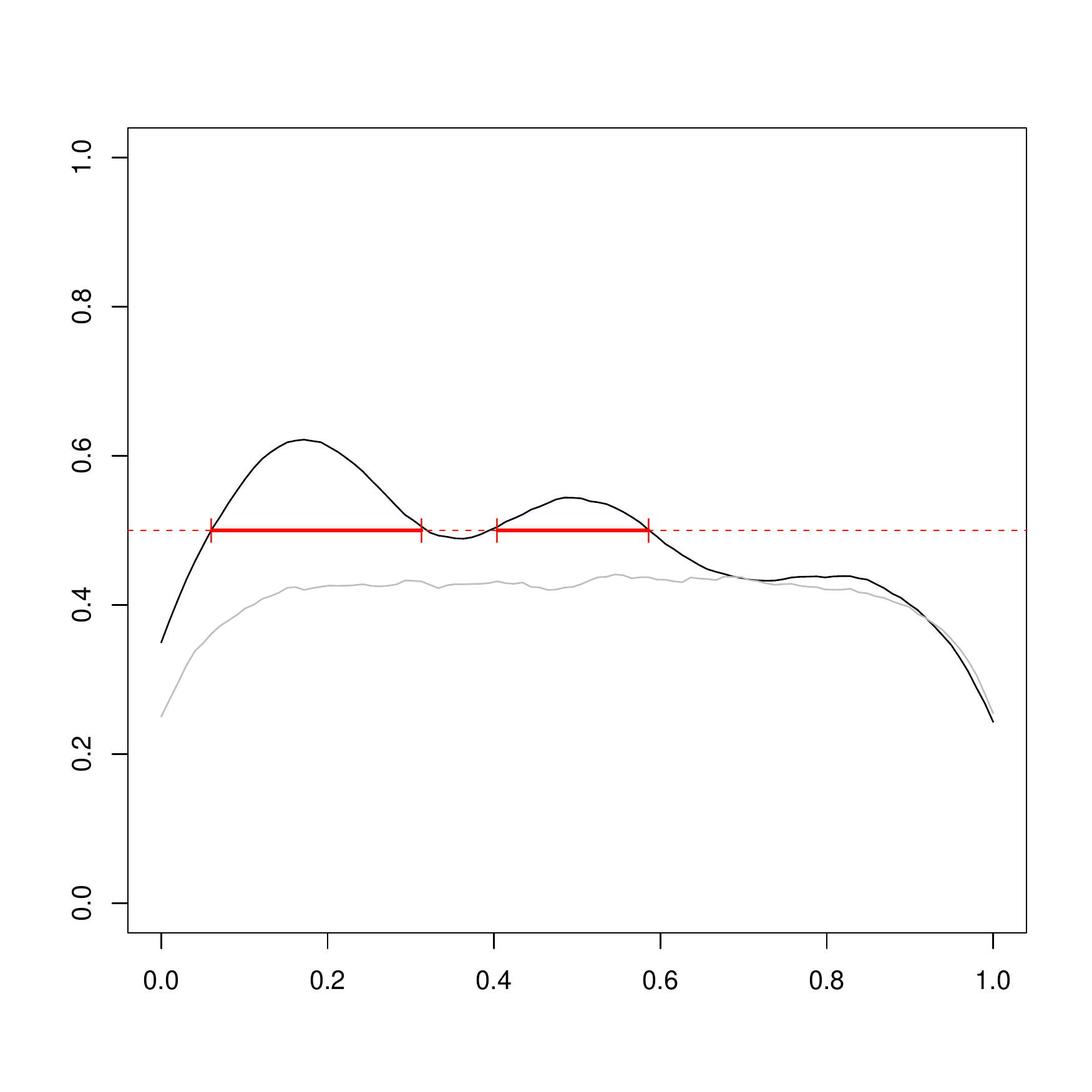}
      \\
    Dataset 1 ($r=5$, $\zeta=1$) & Dataset 3 ($r=5$, $\zeta=1/5$)
  \end{tabular}
  \caption{Prior (in gray) and posterior (in black) probabilities of being in the support computed on Datasets 
    1 and 2. Bayes estimate of support using Theorem~\ref{thm:support}
  with $\gamma=1/2$ are given in red.}\label{fig:support}
\end{figure}

\subsection{Simulation scheme}
\label{sec:scheme}

First of all, we describe how we generate different datasets on which we applied and compared the methods. The support of the covariate curves $x_i$ is $\mathcal{T} = [0,1]$, observed on a regular grid $( t_j )_{j=1,\dots,p}$ on $\mathcal{T}$, for $p=100$. 
We simulate $p$-multivariate Gaussian vectors $x_i$, $i=1,\dots,100$, corresponding to the values of curves $x_i$ for the observation times $(t_j)_j$. The covariance matrix $\Sigma$ of these Gaussian vectors is given by 
$$ \Sigma_{i,j} =\sqrt{\Sigma_{i,i}\Sigma_{j,j}} \exp \left( - \zeta^2 (t_i - t_j)^2  \right), \hspace{0.5cm} \text{ for $i$ and $j$ from 1 to $p$}, $$
where the coefficient $\zeta$ tunes the autocorrelation of the $x_i(t)$. Three different shapes are considered for the functional coefficient $\beta$, given in Figure~\ref{fig:fonction_coef}. 

The first one is a step function, the second one is smooth and is null on small intervals of $\mathcal{T}$ (Smooth), the third one is nonnull only on small intervals of $\mathcal{T}$ (Spiky).
\begin{itemize}
	\itemb Step function: $\beta(t) = 3 \, \mathbf{1}{\{ t \in [0.1,0.3]\}} + 4 \, \mathbf{1}\{ t \in [0.45,0.55]\} -  \mathbf{1}\{ t \in [0.8,0.95]\}$.
	\itemb Smooth: $\beta(t) = 5\times e^{-20(t-0.25)^2} - 2\times e^{ -20(t-0.5)^2 } +2\times e^{-20(t-0.75)^2}$.
	\itemb Spiky: $\beta(t) = 8\times \big( 2 + e^{20-100t} + e^{100t-20}\big)^{-1} - 12\times \big( 2 + e^{60-100t} + e^{100t-60}\big)^{-1} $.
\end{itemize}
The outcomes $y_i$ are calculated according to \eqref{eq:model_bayes_1} with an additional noise following a centred Gaussian distribution with variance $\sigma^2$. The value of $\sigma^2$ is fixed such that the signal to noise ratio is equal to a chosen value $r$.
Datasets are simulated for $\mu=1$ and for the following different values of $\zeta$ and $r$:
\begin{itemize}
	\itemb $\zeta = 1, {1}/{3}, {1}/{5}$, 
	\itemb $r=1,3,5.$
\end{itemize}
Hence, we simulate 27 datasets with different characteristics, that we use in Section~\ref{sec:comparaison} to compare the methods.

\subsection{Performances regarding support estimates}

\begin{table}[!h]
	\caption{Comparison of the support estimate and the support of the Bliss estimate.}
	\begin{center}
	{\scriptsize
	\begin{tabular}{cccccc}
	\hline		
	\multicolumn{3}{c}{} & \multicolumn{2}{c}{Support Error} & Dataset \\
	\hline
	Shape & $r$ & $\zeta$ & Support of the stepwise estimate & Bayes support estimate &  \\
	\hline 
        \multirow{9}{*}{Step function} & 
		$5$ & $1$   & 0.242 & 0.152 & 1 \\ &
		$5$ & $1/3$ & 0.384 & 0.202 & 2 \\ &
		$5$ & $1/5$ & 0.242 & 0.293 & 3 \\ \cdashline{2-6} &
		$3$ & $1$ & 0.232 & 0.091 & 4 \\ &
		$3$ & $1/3$ & 0.323 & 0.394 & 5 \\ &
		$3$ & $1/5$ & 0.424 & 0.465 & 6 \\ \cdashline{2-6} &
		$1$ & $1$ & 0.283 & 0.162 & 7 \\ &
		$1$ & $1/3$ & 0.404 & 0.333 & 8 \\ &
		$1$ & $1/5$ & 0.439 & 0.394 & 9 \\ \cdashline{2-6} \hline
	\end{tabular}
	}
	\\ \vspace{0.2cm}
        \it Section~\ref{sec:scheme} describes the simulation scheme of the datasets.
        Section~\ref{sec:criteria} describes the criteria: Support Error. 
	\end{center}
	\label{tab:support_error}
\end{table}
We begin by assessing the performances of our proposal in term of support recovery. We
focus here on the datasets simulated with the step function as the true coefficient
function. It is the only function among the three functions we have chosen where the real
definition of the support matches with the answer a statistician would expect, see
Figure~\ref{fig:model}. 
The numerical results are given in Table~\ref{tab:support_error}, where we evaluated the
error with the Lebesgue measure of the symmetric difference between the true support $S_0$ and the
estimated one $\widehat{S}$, that is to say $2L_{1/2}(\widehat{S}, S_0)$ with the notation
of Section~\ref{sec:support}.

As we claim at the end of Section~\ref{sec:beta}, the Bayes estimates we have defined in
Theorem~\ref{thm:support} performs much better than relying on the support of a
stepwise estimate of the coefficient function. As also expected the accuracy of the Bayes
support estimate worsens when the autocorrelation within the functional covariate $x_i(t)$
increases. The signal to noise ratio is the second most influent factor that explains the
accuracy of the estimate.

The third interval of the true support, namely $[0.8,0.95]$, is the most difficult to
recover because the true value of the coefficient function over this interval is relatively low
($-1$) compared to the other values ($4$ and $3$) of the coefficient function.
Figure~\ref{fig:support} gives two examples of the posterior probability function
$\alpha(t|\mathcal D)$ defined in Eq.~\eqref{eq:alpha.posterior} where we have highlighted (in red)
the Bayes support estimate with $\gamma=1/2$.  Among these two examples,
the Figure shows that the third interval is recovered
only when there is low autocorrelation in $x_i(t)$ (i.e. Dataset 1). 
Figure~\ref{fig:support} exhibits that the support estimate of
Dataset~1 (low autocorrelation within the covariate) is more
trustworthy than the support estimate of Dataset~3 (high
autocorrelation within the covariate).

For more complex coefficient functions, see Figure~\ref{fig:fonction_coef}, we cannot
compare directly the Bayes support estimate with the true support of the coefficient function that
generated the data. Nevertheless, in the next section, we will compare
the coefficient estimate with the true value of the coefficient function. 

\subsection{Performances regarding the coefficient function}\label{sec:criteria}\label{sec:comparaison}
\begin{figure}
  \centering
  \begin{tabular}{cc}
    \includegraphics[width=.49\textwidth]{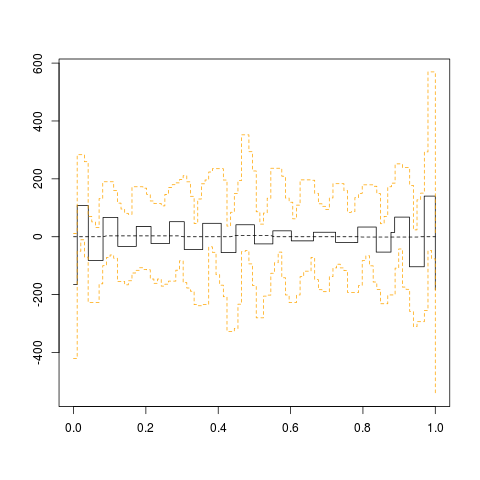}
    &
      \includegraphics[width=.49\textwidth]{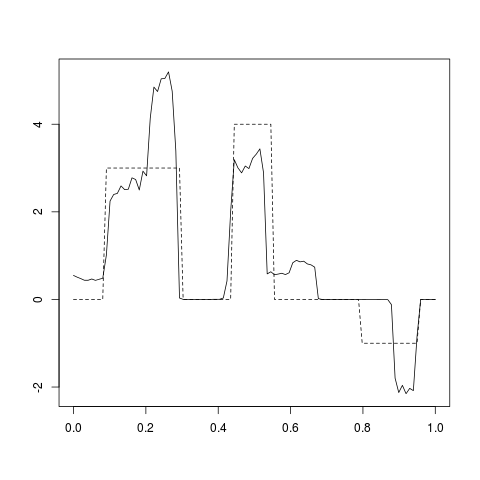}
      \\
    Flirti & Fused Lasso \\
    \\
    \includegraphics[width=.49\textwidth]{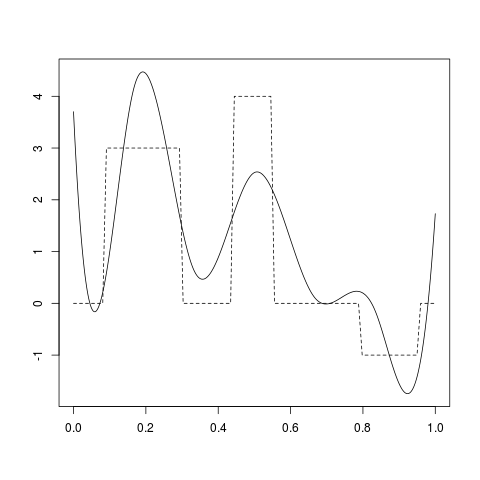}
    &
      \includegraphics[width=.49\textwidth]{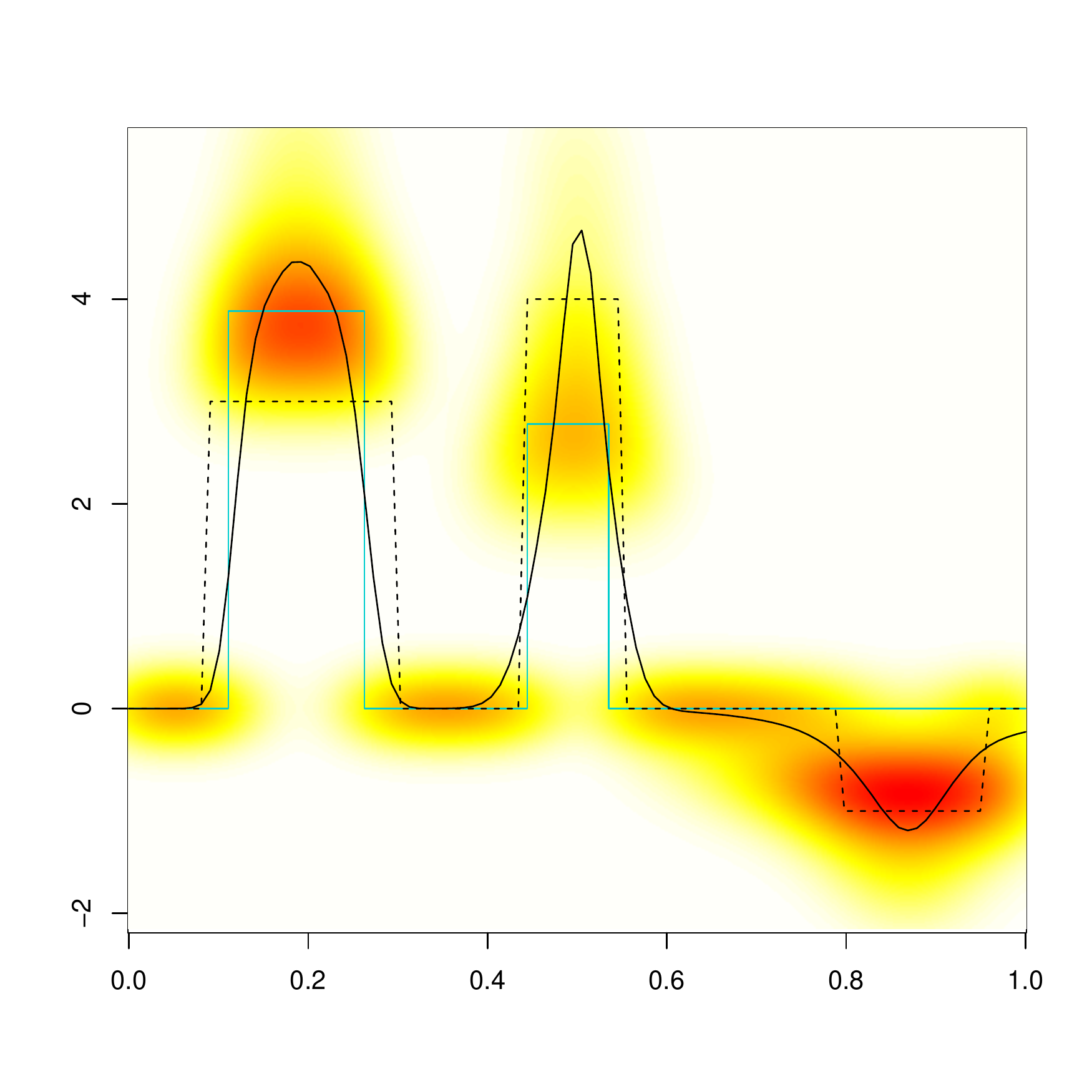}
      \\
    FDA & Bliss
  \end{tabular}
  \caption{Estimates of the coefficient function on Dataset 4 ($r=3$, $\zeta=1$)}
  \label{fig:estim_coeff}
\end{figure}

\begin{figure}
  \centering
  \begin{tabular}{cc}
    \includegraphics[width=.49\textwidth]{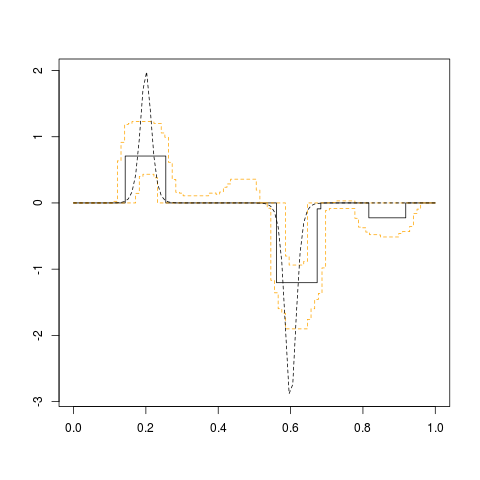}
    &
      \includegraphics[width=.49\textwidth]{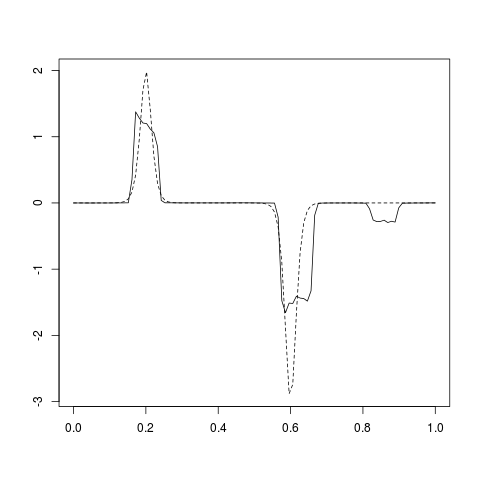}
      \\
    Flirti & Fused Lasso \\
    \\
    \includegraphics[width=.49\textwidth]{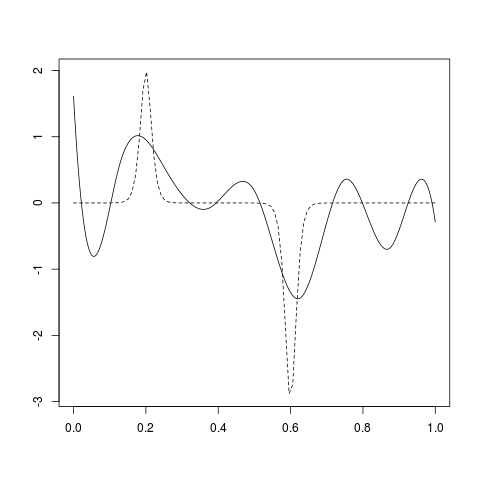}
    &
      \includegraphics[width=.49\textwidth]{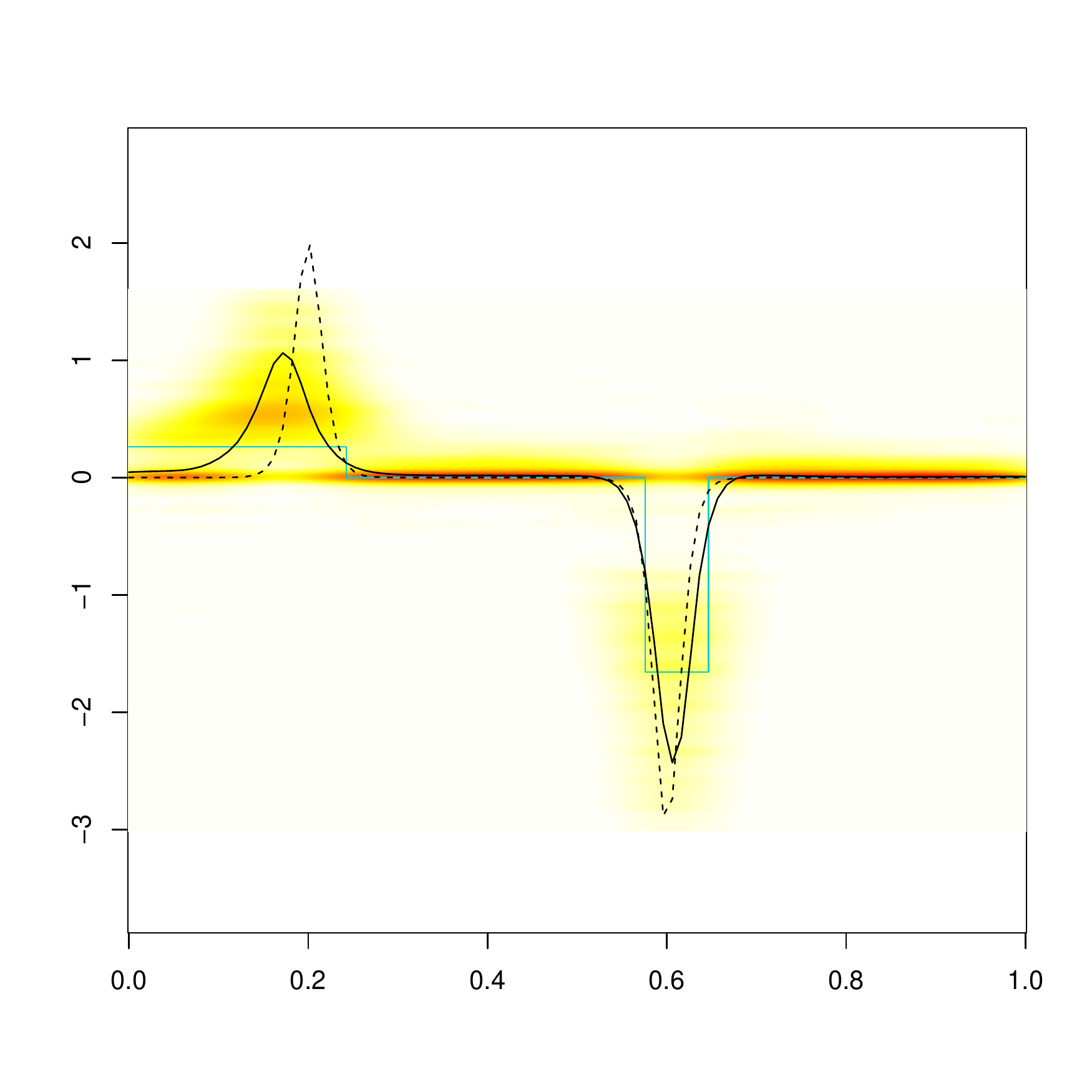}
      \\
    FDA & Bliss
  \end{tabular}
  \caption{Estimates of the coefficient function on Dataset 25 ($r=1$, $\zeta=1$))}
  \label{fig:estim_coeff2}
\end{figure}

In order to compare the methods for the estimation of the coefficient function, we use the
$L^2$-loss, namely
\[\int_0^1 ( \widehat{\beta}(t) - \beta_0(t))^2 \deriv t \]
where $\widehat{\beta}(t)$ is an estimate we compare to the {true} coefficient function $\beta_0(t)$.
Table~\ref{tab:estimate_error} shows the results of Bliss and its
competitors on these simulated datasets.  It appears that the
numerical results of the three methods have the same order of
magnitude. Although the three methods may have different accuracy,
depending on the shape of the coefficient function that generated the
dataset.

Regarding Fused Lasso, we can see in Table~\ref{tab:estimate_error} that its accuracy worsens when the
problem is not sparse, that is to say when the true function is the
``smooth'' function (the red curve of Figure~\ref{fig:fonction_coef}).  Next, we observe
that Flirti is very sensitive. Its numerical results can be sometimes rather accurate, but
sometimes the $L^2$-error can blow up (to exceed $100$) because the method did not manage
to tune its parameters.  The $L^2$-Bliss estimate defined in
Proposition~\ref{pro:1} frequently overperforms the other methods. This
first conclusion is not surprising because the $L^2$-Bliss estimate has been defined to
optimize the $L^2$-loss integrated over the posterior distribution. 

Even in situations where the true function is stepwise, the stepwise Bliss estimate of
Proposition~\ref{pro:3} is less accurate than the $L^2$-Bliss estimate, except for two
examples (datasets 6 and 9). Nevertheless we do argue that the stepwise Bliss estimate was built to provide a
trade off between accuracy regarding the support estimate and accuracy regarding the
coefficient function estimate. Thus the stepwise estimate is a
balance between support estimate and coefficient function estimate
that can help the statistician who can then get an
interpretation of the underlying phenomena that generated the data. In
other words, the stepwise-Bliss estimate is not the best neither at
estimating the support nor at approximating the coefficient function,
but provides a tradeoff.

To show more detailed results we have presented the estimate of the coefficient
function in two cases.
\begin{itemize}
\item Figure~\ref{fig:estim_coeff} displays the numerical results on
  Dataset 4 (medium level of signal, low level of autocorrelation with
  the covariates). As can
  be expected when the true coefficient is a stepwise function, the stepwise
  Bliss estimate behaves nicely. The representation of the marginals of the posterior
  distribution with a heat map shows the confidence we can have in the
  Bayes estimate of the coefficient function. The
  smooth estimate nicely follows the regions of high posterior density.
  Here, the stepwise estimate clearly highlights two time periods 
  (the first two intervals of the true support) and the sign of the coefficient
  function on these intervals. We can compare our proposal with its competitors. Flirti did not manage to tune its
  own parameters, and the Flirti estimate is completely
  irrelevant. Fused Lasso on a discretized version of the functional
  covariate provides a relatively nice estimate of the coefficient
  function. And FDA is not that bad, although the estimate is clearly
  too smooth to match the true coefficient function.
\item Figure~\ref{fig:estim_coeff2} displays the numerical results on
  Dataset 25 (low level of signal, and low level of autocorrelation
  within the covariates). In this example, the true coefficient is not
  stepwise, but smooth, and is around zero on large time periods. The $L^2$-Bliss estimate of Proposition~\ref{pro:1} matches
  approximately the true coefficient function. The stepwise-Bliss
  estimate is a little bit poorer (maybe because of the difficult
  calibration of the simulated
  annealing algorithm). When comparing these results with other
  estimates on this dataset, we see that Flirti and Fused Lasso
  performed  also decently, even if they both highlight a third
  time period (around $t=0.85$) where they infer a negative
  coefficient function instead of $0$. Flirti is at its best here, and
  has obviously managed to tune its own parameters in a relevant way.
  The confidence bands of Flirti are then
  reliable, but we stress here that they are relatively wide around
   periods where
  the Flirti estimate is null and does not reflect high confidence in any support estimate
  based on Flirti. 
  Finally, the comments on FDA are the same as Dataset 4, the FDA
  estimate is clearly too smoothed to match the true coefficient
  function. 
\end{itemize}


\begin{table}
	\caption{{Numerical results of Bliss, Flirti, Fused lasso and FDA on
      the Simulated Datasets. 
    }}
	\begin{center}
	{\scriptsize
	\begin{tabular}{ccccccccc}
	\hline		
	\multicolumn{3}{c}{} & \multicolumn{5}{c}{\bf$L^2$-error} & \bf Dataset \\
	Shape & $r$ & $\zeta$ & stepwise Bliss & $L^2$-Bliss & Fused lasso & Flirti & FDA &  \\
	\hline 
	\multirow{9}{*}{Step function} & 
		$5$ & $1$   & 1.126 & 0.740 & \textbf{0.666} & 1.288  & 1.514 &1 \\ &
		$5$ & $1/3$ & 2.221 & \textbf{1.415} & 1.947 & 1.781  & 1.997 & 2 \\ &
		$5$ & $1/5$ & 2.585 & \textbf{1.656} & 1.777 & 3.848  & 1.739 & 3 \\ \cdashline{2-9} &		
		$3$ & $1$   & 1.283 & \textbf{0.821} & 0.984 & $10^3$ & 1.203 & 4 \\ &		
		$3$ & $1/3$ & 1.531 & \textbf{1.331} & 1.936 & $10^4$ & 1.830 & 5 \\ &		
		$3$ & $1/5$ & 2.266 & 2.989 & 2.036 & \textbf{1.772}  & 2.144 & 6 \\ \cdashline{2-9} &		
		$1$ & $1$   & 1.589 & \textbf{0.747} & 0.995 & 3.848  & 1.577 & 7 \\ &		
		$1$ & $1/3$ & 2.229 & \textbf{1.817} & 2.214 & $10^4$ & 2.307 & 8 \\ &		
		$1$ & $1/5$ & \textbf{1.945} & 2.364 & 2.028 & 3.848  & 4.437  & 9 \\ \cdashline{2-9}
    \multirow{9}{*}{Smooth} & 
		$5$ & $1$   & 0.510 & \textbf{0.134} & 0.601 & 0.166  & 0.573 & 10 \\ &		
		$5$ & $1/3$ & 0.807 & 0.609 & \textbf{0.442} & 2.068  & 1.103 & 11 \\ &		
		$5$ & $1/5$ & 1.484 & \textbf{1.352} & 2.325 & 2.068  & 1.650 & 12 \\ \cdashline{2-9} &		
		$3$ & $1$   & 0.776 & 0.416 & 0.320 & \textbf{0.263}  & 3.295 & 13 \\ &		
		$3$ & $1/3$ & \textbf{0.855} & 0.954 & 6.790 & 2.068  & 1.819 & 14 \\ &		
		$3$ & $1/5$ & 1.291 & \textbf{1.162} & 1.742 & 1.328  & 1.759 & 15 \\ \cdashline{2-9} &		
		$1$ & $1$   & 0.932 & 0.641 & 0.652 & 2.335  & \textbf{0.616} & 16 \\ &		
		$1$ & $1/3$ & 0.719 & \textbf{0.283} & 0.613 & $10^4$ & 1.308 & 17 \\ &		
		$1$ & $1/5$ & 1.536 & \textbf{1.006} & 4.680 & 5.430  & 2.985 & 18 \\ \cdashline{2-9}
    \multirow{9}{*}{Spiky} & 
		$5$ & $1$   & 0.099 & \textbf{0.013} & 0.059 & 0.035  & 0.239 & 19 \\ &		
		$5$ & $1/3$ & 0.208 & \textbf{0.144} & 0.260 & 0.271  & 0.349 & 20 \\ &		
		$5$ & $1/5$ & 0.285 & 0.251 & \textbf{0.181} & 0.226  & 0.306 & 21 \\ \cdashline{2-9} &		
		$3$ & $1$ & 0.187 & \textbf{0.023} & 0.638 & 0.136  & 0.584 & 22 \\ &		
		$3$ & $1/3$ & 0.257 & 0.202 & \textbf{0.159} & 0.277  & 0.258 & 23 \\ &		
		$3$ & $1/5$ & 0.269 & \textbf{0.260} & 0.459 & 0.276  & 1.050 & 24 \\ \cdashline{2-9} &		
		$1$ & $1$ & 0.144 & \textbf{0.087} & 0.123 & 0.166  & 0.270 & 25 \\ &		
		$1$ & $1/3$ & 0.242 & \textbf{0.223} & 0.260 & $10^2$ & 0.270 & 26 \\ &		
		$1$ & $1/5$ & 0.273 & 0.279 & \textbf{0.221} & 0.301  & 0.405 & 27 \\ \cdashline{2-9}
	\end{tabular}
	}\end{center}
        \flushleft \footnotesize
        Section~\ref{sec:scheme} describes the simulation scheme of the datasets.
        The stepwise Bliss estimate is the estimate defined in Proposition~\ref{pro:3},
        while the $L^2$-estimate is the smooth estimate defined in Proposition~\ref{pro:1}.
	\label{tab:estimate_error}
\end{table}

\subsection{Tuning the hyperparameters}
\label{choixhyperparametres}	\label{sec:sensitive}

We can now discuss our recommandation on the hyperparameters of the
model, given at the end of Section~\ref{sec:model}.
For this study, we applied our methodology on Dataset~1 and fixed the hyperparameters
$v_0$, $v$, $a$ around the recommended values.
We recall that Dataset~1 is a synthetic dataset simulated with a coefficient
function that is a step function (the black curve of
Figure~\ref{fig:fonction_coef}), with a high level of signal over noise ($r=5$)
and with a low level of autocorrelation within the covariates($\zeta=1$). The following values
are considered for each hyperparameter:
\begin{itemize}
	\item for $a$: $0.5 / K$, $0.2 / K$, $0.1 / K$, $0.07 / K$ and $0.05 / K$;
	\item for $v$: 10, 5, 2, 1 and 0.5;
        \item and for $K$: any integer between $1$ and $10$.
\end{itemize}
The numerical results are given in Table~\ref{tab:sensitivity}. The default values we recommend are
not the best values here, but we have done numerous other trials on
many synthetic datasets and these choices are
relatively robust. We do not highlight any particular value for $K$ since this value can (and
should) be chosen with the Bayesian model choice machinery.

\begin{table}[!h]
	\caption{{Performances of Bliss with respect
            to the tuning of the hyperparameters.} }	
	\begin{center}
	{\scriptsize
	\begin{tabular}{ccccc}
	\hline		
	\multicolumn{1}{c}{} & \multicolumn{2}{c}{Error on the $\beta$} 
          & \multicolumn{2}{c}{ Error on the support} \\
	\hline
	 & stepwise-Bliss & $L^2$-Bliss & Support of the stepwise estimate & Bayes support estimate    \\
	\hline 
		$a = 0.5 / K$              & 1.000 & 0.698 & 0.222 & 0.439  \\ 
		$a = 0.2 / K$ $\heartsuit$ & 1.013 & 1.135 & 0.222 & 0.192  \\ 
		$a = 0.1 / K$              & 1.642 & 1.364 & 0.242 & 0.202 \\ 
		$a = 0.07 / K$             & 3.060 & 1.645 & 0.364 & 0.212  \\ 
		$a = 0.05 / K$             & 2.032 & 1.888 & 0.263 & 0.263  \\ \cdashline{1-5} 
		$v = 10$             & 1.628 & 1.125 & 0.242 & 0.192  \\ 
		$v = 5$ $\heartsuit$ & 1.711 & 1.131 & 0.242 & 0.192  \\   		                         
		$v = 2$              & 1.082 & 1.143 & 0.273 & 0.192  \\ 
		$v = 1$              & 1.207 & 1.119 & 0.273 & 0.192  \\ 
		$v = 0.5$            & 1.675 & 1.129 & 0.263 & 0.192  \\ \cdashline{1-5}  
		$K = 1$               & 1.798 & 1.782 & 0.424 & 0.449 \\ 
		$K = 2$               & 0.993 & 1.101 & 0.222 & 0.222 \\ 
		$K = 3$               & 1.696 & 1.124 & 0.242 & 0.192 \\ 		                         
		$K = 4$               & 1.736 & 1.159 & 0.283 & 0.172 \\ 
		$K = 5$               & 2.081 & 1.233 & 0.303 & 0.172 \\ 
		$K = 6$               & 2.177 & 1.243 & 0.283 & 0.202 \\ 
		$K = 7$               & 2.135 & 1.221 & 0.303 & 0.232 \\ 
		$K = 8$               & 1.343 & 1.184 & 0.263 & 0.242 \\ 
		$K = 9$               & 1.439 & 1.166 & 0.263 & 0.328 \\ 
		$K = 10$              & 1.897 & 1.089 & 0.364 & 0.348 \\ 	 
	 \hline
	\end{tabular}	
	}
	\\ \vspace{0.2cm}
        \it  The $\heartsuit$ symbol indicates the default values.
	\end{center}
	\label{tab:sensitivity}
\end{table}

\section{Application to the black P\'erigord truffle dataset} 
\label{sec:reel}
We apply the Bliss method on a dataset to predict the amount of production of black truffles given the rainfall curves.
\begin{figure}	
	\centering
	\includegraphics[height=6cm]{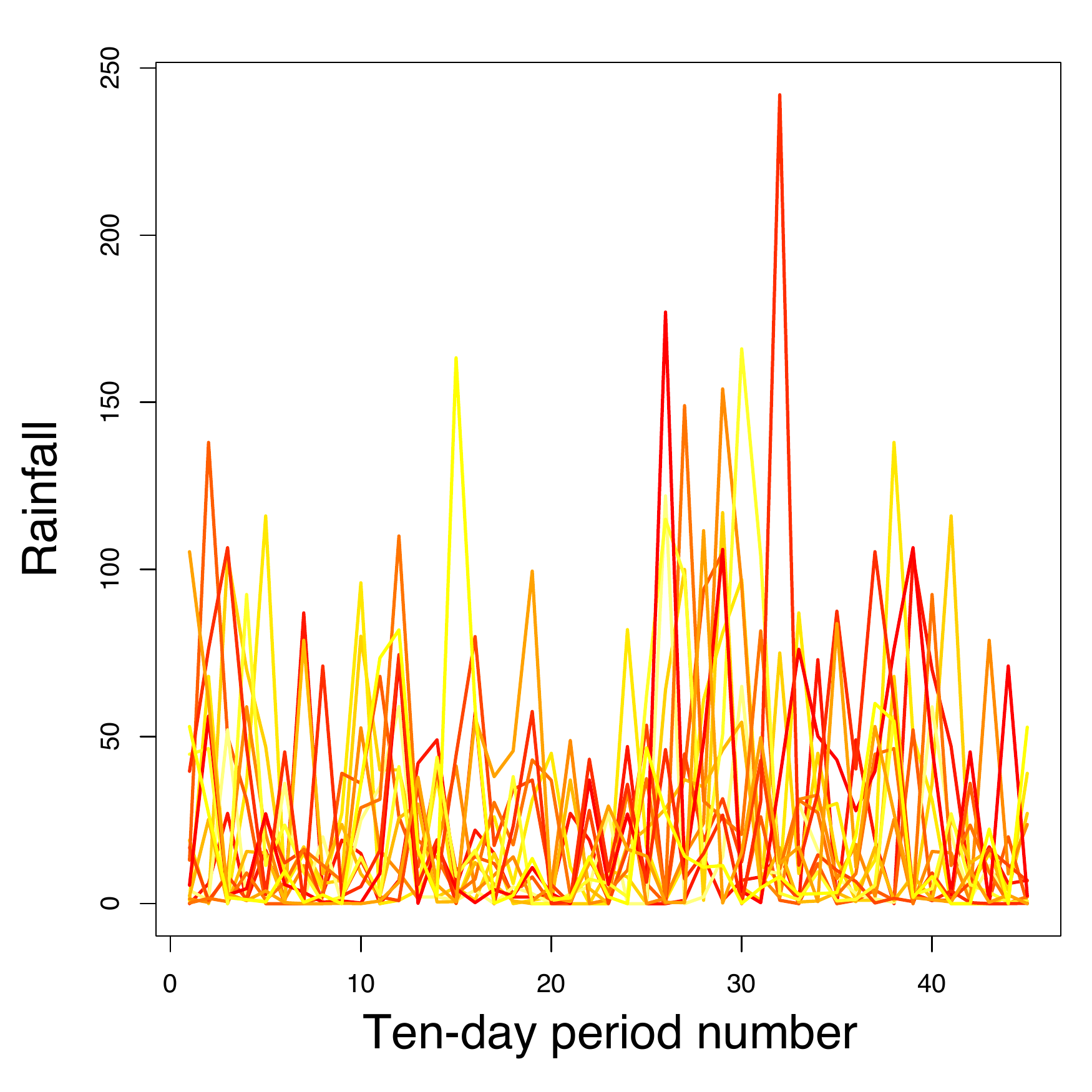}~\includegraphics[height=6cm]{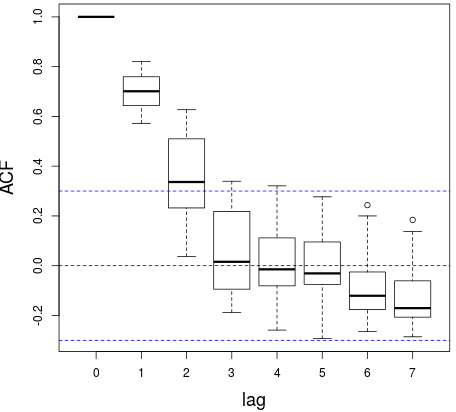} \\
	\caption{{Rainfall of the Truffle datsaset.} Left:{\it Plot
          shows the rainfall for each year, color-coded by their
          truffle yield.} Right:{\it Autocorrelation of the 13
          observed rainfall covariates, with lag in number of ten-day periods.} }
	\label{fig:donnees}
\end{figure}
The black P\'erigord truffle (\textit{Tuber Melanosporum Vitt.}) is
one of the most famous and valuable edible mushrooms, because of its
excellent aromatic and gustatory qualities. It is the fruiting body of
a hypogeous Ascomycete fungus, which grows in ectomycorrhizal
symbiosis with oaks species or hazelnut trees in Mediterranean
conditions. Modern truffle cultivation  involves the plantation of
orchards with tree seedlings inoculated with \textit{Tuber
  Melanosporum}. The planted orchards could then be viewed as
ecosystems that should be managed in order to favour the formation and
the growth of truffles. The formation begins in late winter with the
germination of haploid spores released by mature ascocarps. Tree roots
are then colonised by haploid mycelium to form ectomycorrhizal
symbiotic associations. Induction of the fructification (sexual
reproduction) occurs in May or June (the smallest truffles have been
observed in mid-June). Then the young truffles grow during summer
months and are mature between  the middle of November and the middle
of March (harvest season). The production of truffles should then be
sensitive to climatic conditions throughout the entire year
\citep{LeTacon2014}. However, to our knowledge few studies focus on
the influence of rainfall or irrigation during the entire year
(\citealp{Demerson2014,LeTacon2014}). Our aim is then to investigate
the influence of rainfall throughout the entire year on the production
of black truffles. Knowing this influence could lead to a
better management of the orchards, to a better understanding of the
sexual reproduction, and to a better understanding of the effects of
climate change. Indeed, concerning sexual reproduction,
\citet{LeTacon2014,LeTacon2016} made the assumption that climatic
conditions could be critical for the initiation of sexual reproduction throughout the development of the mitospores expected to occur in late winter or spring. And concerning climate change, its consequences on the geographic distribution of truffles is of interest (see \citealp{Splivallo2012} or \citealp{Buntgen2011}, among others).

The analyzed data were provided by J. Demerson. They consist of the
rainfall records on an orchard near Uz\`es (France) between 1985 and
1999, and of the production of black truffles on this orchard between
1985 and 1999. In practice, to explain the production of the year $n$,
we take into account the rainfall between the 1st of January of the
year $n-1$ and the 31st of March of the year $n$. Indeed, we want to
take into account the whole life cycle, from the formation of new
ectomycorrhizas following acospore germination during the winter
preceding the harvest (year $n-1$) to the harvest of the year $n$. The
cumulative rainfall is measured every 10 days, hence between the 1st
of January of the year $n-1$ and the 31st of March of the year $n$ we
have the rainfalls associated with 45 ten-day periods, see Figure~\ref{fig:donnees}. This dataset can be considered as reliable, as the rainfall records have been made exactly on the orchard, and the orchard was not irrigated. 

\paragraph*{Biological assumptions at stake} From the literature we
can spotlight the following periods of time which might influence the
growth of truffles.
\begin{itemize}
\item[Period \#1:] Late spring and summer of year $n-1$. This is the (only) period for
  which all experts are unanimous to say it has a particular effect. \citet{Buntgen2012a}, \citet{Demerson2014} or \citet{LeTacon2014} all confirm the
importance of the negative effect of summer hydric deficit on truffle production: they found it
to be the most important factor influencing the production. Indeed, in summer the truffles need
water to survive the high temperatures and to grow. Otherwise they can
dry out and die.
\item[Period \#2:] Late winter of year $n-1$, as shown by \citet{Demerson2014} and
\citet{LeTacon2014}. Indeed, as explained in \citet{LeTacon2014}, consistent water availability
in late winter could support the formation of new mycorrhizae, thus allowing a new
cycle. Moreover, from results obtained by \citet{Healy2013} they made the assumption that
rainfall is critical for the initiation of sexual reproduction throughout development of
mitospores, which is expected to occur in late winter or spring of the year $n-1$. This is an
assumption as the occurrence and the initiation of sexual reproduction is largely unknown, see
\citet{Murat2013} or \citet{LeTacon2016}.
\item[Period \#3:] November and December of year $n-1$, as claimed by
  \citet{Demerson2014} and \citet{LeTacon2014}. Le Tacon et
  al. explained that rainfall in autumn allows the growth of young
  truffles which have survived the summer. 
\item[Period \#4:] September of year $n-1$, as claimed by
  \citet{Demerson2014}. Excess water in this period should be harmuful to truffles. The assumption made was that in September the
  soil temperature is still high, so micro-organisms responsible for
  rot are quite active, while a wet truffle has its
  respiratory system disturbed and can not defend itself against these
  micro-organisms.  
\end{itemize}
The challenge is to confirm some of these periods with Bliss, despite
the small size of the dataset. In particular, each rainfall curve is
discretized with only 45 points (cumulative rainfall every 10 days)
and we have at our disposal only 13 observations. 

\begin{figure}[ht]
  \centering
  \includegraphics[height=6cm]{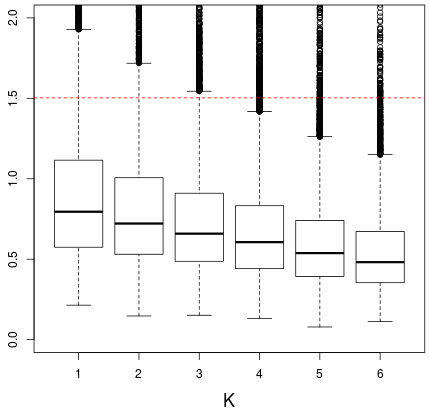}~\includegraphics[height=6cm]{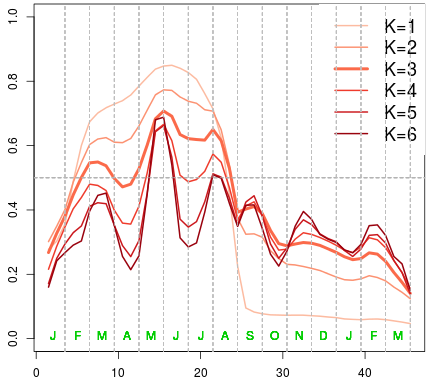}
  \caption{Sensitivity of Bliss to the value of $K$ on the truffle
    dataset. Left: {\it Boxplot of the posterior distribution of the
      variance of the error,
      $\sigma^2$, compared to the variance of the output $y$ (red
      dashed line)}. Right: {\it Posterior probability
    $\alpha(t|\mathcal D)$ for different values of $K$.}}
  \label{fig:sensitivity_K}
\end{figure}

\paragraph*{Running Bliss} As explained above (in Section 3.2), part of the difficulty of
the inference problem comes from autocorrelation within the
covariate. Figure~\ref{fig:donnees} shows that the autocorrelation can
be considered as null when the lag is 3 or more in number of ten-day periods. In other
words the rainfall background presents autocorrelation within a period of
time of about a month (keeping in mind that the whole history we
consider lasts 15 months).

The first and maybe most important hyperparameter is $K$, the number
of intervals in the coefficient functions from the prior. Because of
the discretization of the rainfall, and the number of observations,
the value of $K$ should stay small to remain parsimonious. 
Because of the size of the dataset, we have set the hyperparameter $a$
to obtain a prior probability of being in the support of about
$0.5$. The results are given in Figure~\ref{fig:sensitivity_K}. As can
be seen on the left of this Figure, the error variance $\sigma^2$ decreases
when $K$ increases, because models of higher dimension can more easily
fit the data. The main question is when do they overfit the data?
Looking at the right panel of Figure~\ref{fig:sensitivity_K}, we can
consider how the posterior probability $\alpha(t|\mathcal D)$ depends
on the value of $K$ and choose a reasonable value. First, for $K=1$ or
$2$, the posterior probability is high during a first long
period time until August of year $n-1$ and falls to  much lower values
after that. Thus, these small values of $K$ provide a rough picture of
dependency. Secondly, for $K=4, 5$ or $6$, the posterior probability
$\alpha(t|\mathcal D)$ varies between $0.2$ and $0.7$ and shows
doubtful variations after November of year $n-1$ and other strong
variations during the summer of year $n-1$ that are also doubtful. Hence we decided to
rely on $K=3$ although this choice is rather subjective.

\paragraph*{Conclusions on the truffle dataset}
We begin by noting that about half of the variance of the output (the amount of
production of truffle) is explained by the rainfall given the posterior
distribution of $\sigma^2$ in the left panel of Figure~\ref{fig:sensitivity_K}.
The support estimate $\widehat{S}_{0.5}(\mathcal D)$ with $K=3$ is composed
of two disjoint intervals: a first one from May of year $n-1$ to the
second ten-day period  of August with the highest posterior probability, and a
second one from the third ten-day period of February of year $n-1$ to the end
of March of year $n-1$ with a smaller posterior probability. Thus, as
far as we can tell from this analysis, Periods \#1 and \#2 are
validated by the data. Period \#3 cannot be validated although
the posterior probability $\alpha(t|\mathcal D)$ presents small bumps
around theses periods of time for highest values of $K$. For $K=3$,
the value of $\alpha(t|\mathcal D)$ stays around $0.3$ on Period \#3. Finally,
regarding Period \#4, we can see a small bump on the curve
$\alpha(t|\mathcal D)$ around this period of time even for $K=3$, but
the highest value of the posterior probability on this period is about
$0.4$. Hence we chose to remain undecided on Period \#4.

\section{Conclusion}
In this paper, we have provided a full Bayesian methodology to analyse linear models with
time-dependent functional covariates. The main purpose of our study was to estimate of the
support of the coefficient function to search the periods of time which influences the most the
outcome. We rely on piecewise constant coefficient functions to set the prior, which has
four benefits. The first benefit is parsimony of the Bliss model, which turns two thirds of the
parameter's dimension to the estimation of the support. The second benefit with our
Bayesian setting that begins by defining the support is that we can rely on the ridge-Zellner
prior to handle the autocorrelation within the functional covariate. This fact sets Bliss apart
from Bayesian methods relying on spike-and-slab prior to handle sparsity. The third benefit is
avoiding cross-validation to tune internal parameters of the method. Indeed, cross-validation
methods optimize the performance regarding the model's predictive power, and not the accuracy
of the support estimate. And, last but not least, the fourth benefit is the ability to compute
numerically the posterior probability that a given date is in the support,
$\alpha(t|\mathcal D)$, whose value gives a clear hint on the reliability of the support
estimate. Nevertheless a serious limitation of our Bayesian model is that it can handle only
covariate functions of one variable (we call time in the paper). Indeed the shape of the
support of a function of more than one variable is much more complex than an union of intervals
and cannot be easily modelled in a nonparametric, but parsimonious manner.

We have provided numerical results regarding the power of Bliss on a bunch of synthetic
datasets as well as a dataset studying the black P\'erigord truffle.  We have shown by
presenting some of these examples in details how we can interpret the results of Bliss, in
particular how we can rely on the posterior probabilities $\alpha(t|\mathcal D)$ or the heatmap
of posterior distribution of the coefficient function to assess the reliability of our
estimates. Bliss provides two main outputs: first an estimate of the support of the coefficient
function without targeting the coefficient function, and second a trade-off between support
estimate and coefficient function estimate through the stepwise estimate of Proposition~\ref{pro:3}.
Moreover our prior can straightforwardly be encompassed into a linear model with other functional or scalar
covariates.


\bibliography{bib/BLiSS}

\section*{Acknowledgement}
  We are very grateful to Jean Demerson for providing the truffle dataset and for his
  explanations. Pierre Pudlo carried out this work in the framework of the Labex Archim\`ede
  (ANR-11-LABX-0033) and of the A*MIDEX project (ANR-11-IDEX-0001-02), funded by the
  ``Investissements d'Avenir'' French Government program managed by the French National Research
  Agency (ANR).

\section*{Supplementary Materials}
The implementation of the method is available at the following
webpage:\\
\url{http://www.math.univ-montp2.fr/~grollemund/Implementation/BLiSS/}.

\appendix
\section{Theoretical results}
\label{sec:Theoretical}

\subsection{Proof of Theorem~\ref{thm:support}}
\label{sec:thm1}
Without loss of generality we can assume that $\mathcal T=[0;1]$.
We begin the proof with the following lemma whose simple proof is left
to the reader.
\begin{lemma} \label{lemma}
  Set $\psi^\ast(\gamma, \alpha)=\min\{\gamma(1-\alpha)\, ;\,
  (1-\gamma)\alpha\}$ for any $\alpha, \gamma\in[0;1]$. We have
  \[
  \psi^\ast(\gamma,\alpha) = \begin{cases} \gamma(1-\alpha) & \text{if }
      \gamma\le \alpha, \\  (1-\gamma)\alpha &\text{if }\gamma \ge \alpha.\end{cases}
  \]
\end{lemma}
Recall that the posterior loss we optimise is given in
\eqref{eq:optim.support}, where $S$ is any Borel subset of $\mathcal T=[0;1]$.
Using Fubini's theorem (for non-negative functions) and the definition
of $\alpha(t|\mathcal D)$ given in \eqref{eq:alpha.posterior}, we have
\begin{align}
  \int_{\Theta_K} L_\gamma(S, S_\theta) \pi_K(\theta|\mathcal D)\deriv
  \theta 
  &= \gamma\int_0^1 \int_{\Theta_K}\mathds{1}\{t\in S\setminus
    S_\theta\} \pi_K(\theta|\mathcal D)\deriv\theta \deriv t \notag
\\ & \quad
    + (1-\gamma)\int_0^1 \int_{\Theta_K}\mathds{1}\{t\in S_\theta\setminus
    S\} \pi_K(\theta|\mathcal D)\deriv\theta \deriv t \notag
    \\
  & = \int_0^1 \psi_S\big(t,\gamma, \alpha(t|\mathcal D)\big)\deriv t
    \label{eq:pr1}
\end{align}
where, for all $\alpha\in[0;1]$ we have set
\[
\psi_S(t,\gamma, \alpha) = \mathds{1}\{t\in S\}  \gamma\big(1-\alpha\big)
    + \mathds{1}\{t\not\in S\} (1-\gamma) \alpha.
\]
Now, whatever the set $S$, $\psi_S(t,\gamma, \alpha) \ge
\psi^\ast(\gamma, \alpha)$. Reporting this bound in \eqref{eq:pr1} yields
\[
\int_{\Theta_K} L_\gamma(S, S_\theta) \pi_K(\theta|\mathcal D)\deriv
  \theta  \ge \int_0^1 \psi^\ast\big(\gamma, \alpha(t|\mathcal D) \big)\deriv t
\]
whatever the Borel set $S$.  Moreover, this inequality is an equality
if and only if the Borel set $S$ is chosen so that, for almost all
$t\in[0;1]$,
$\psi_S\big(t,\gamma, \alpha(t|\mathcal D)\big) = \psi^\ast\big(\gamma,
\alpha(t|\mathcal D)\big)$. Using Lemma~\ref{lemma}, the last condition
is equivalent to saying that for almost all $t\in[0;1]$, either
$\alpha(t|\mathcal D)=\gamma$ or $\big(t\in S \iff \gamma \le\alpha(t|\mathcal D)\big)$.
This concludes the proof of Theorem~\ref{thm:support}.
\hfill\qed

\subsection{Proof of Proposition~\ref{pro:1}}
\label{sec:proof1}
  Obviously, $\widehat{\beta}_{L^2}(\cdot)$ minimizes
  \[
  \int \int_{\mathcal T} \left(\beta_\theta(t) - d(t)\right)^2 \deriv t\ \pi_K(\theta|\mathcal D)\deriv \theta = 
  \int_{\mathcal T} \int \left(\beta_\theta(t) - d(t)\right)^2 \pi_K(\theta|\mathcal D)\deriv \theta\ \deriv t
  \]
  because it does optimize $\int \left(\beta_\theta(t) - d(t)\right)^2 \pi_K(\theta|\mathcal D)\deriv \theta$ for all $t\in\mathcal T$. It remains to show that $\widehat{\beta}_{L^2}(\cdot)\in L^2(\mathcal T)$.  
  We have
  \begin{align*}
    \|\widehat{\beta}_{L^2}(\cdot)\|^2 & = \int_{\mathcal T} \left(\int\beta_\theta(t)\pi_K(\theta|\mathcal D)\deriv \theta\right)^2\deriv t
    \\
     & = \int_{\mathcal T}\iint \beta_\theta(t)\beta_{\theta'}(t)\ \pi_K(\theta|\mathcal D) \pi_K(\theta'|\mathcal D) \deriv\theta \deriv\theta' \deriv t
    \\
    & = \iint \int_{\mathcal T} \beta_\theta(t)\beta_{\theta'}(t)\deriv t\ \pi_K(\theta|\mathcal D) \pi_K(\theta'|\mathcal D) \deriv\theta \deriv\theta'
    \\
    & \le \iint \|\beta_\theta(\cdot)\| \|\beta_{\theta'}(\cdot)\|\ \pi_K(\theta|\mathcal D)\pi_K(\theta'|\mathcal D) \deriv\theta \deriv\theta' 
      \quad \text{with Cauchy-Schwarz inequality}
    \\
    & \le \left(\int \|\beta_\theta(\cdot)\| \pi_K(\theta|\mathcal D) \deriv\theta \right)^2
  \end{align*}
  And the last integral is finite because of the assumption. Hence $\widehat{\beta}_{L^2}(\cdot)$ is in $L^2(\mathcal T)$. \hfill \qed
\subsection{Proof of Proposition~\ref{pro:3}}
\label{sec:proof2}
First, the norm $\| d(\cdot) - \widehat{\beta}_{L^2}(\cdot)\|$ is non negative, hence the set
\[
\Big\{\| d(\cdot) - \widehat{\beta}_{L^2}(\cdot)\|,\ d(\cdot) \in \mathcal E_{K_0}^{\eps} \Big\}
\]
admits an infimum. Let $m$ denote this infimum. We have to prove that $m$ is actually a minimum of the above set, namely that there exists a function  $d(\cdot) \in \mathcal E_{K_0}^{\eps}$ such that
$m = \| d(\cdot) - \widehat{\beta}_{L^2}(\cdot)\|$.

\bigskip

To this end, we introduce a minimizing sequence $\{ d_n(\cdot)\}$ and
we will show that one of its subsequence admits a limit within $\mathcal E_{K_0}^{\eps}$.
Let $d_n(\cdot)$ be such that

\begin{equation}
m= \inf \Big\{\| d(\cdot) - \widehat{\beta}_{L^2}(\cdot)\|,\ d(\cdot)
\in \mathcal E_{K_0}^{\eps} \Big\}
\le \| d_n(\cdot) - \widehat{\beta}_{L^2}(\cdot)\|  \le m +
2^{-n}.\label{eq:mmmm}
\end{equation}
The step function $d_n(\cdot)$ can be written as 
\[
d_n(t) = \sum_{k=1}^L \alpha_{k,n} \mathds 1\{t \in (a_{k,n},\ b_{k,n})\}
\]
where the $(a_{k,n},\ b_{k,n})$, $k=1,\ldots,L$ are non overlapping
intervals. Note that their number $L$ does not depend on $n$ because
all $d_n(\cdot)$ lie in $\mathcal E_{K_0}$ for some fixed value of $K_0$,
and we can always choose $L= 2K_0-1$. Moreover, because $d_n(t)$ is in
$\mathcal F^{\eps}$, we can assume that 
\begin{equation} \label{eq:epskn}
b_{k,n}-a_{k,n}\ge \eps, \quad \text{for all }k,n.
\end{equation}

\bigskip

Now the sequence $\{a_{1,n}\}_n$ has its elements in the compact
interval $\mathcal T$ hence we extract a subsequence (still denoted
$\{a_{1,n}\}_n$) which converges an element $a_{1,\infty}$ of $\mathcal
T$.
Likewise, by extracting subsequences $2L$ times, we can assume that all
sequences $\{a_{1,n}\}_n$,\ldots, $\{a_{L,n}\}_n$, $\{b_{1,n}\}_n$,
\ldots, $\{b_{L,n}\}_n$ are convergent, and that
\[
a_{k,\infty} = \lim_{n\to\infty} a_{k,n}, \quad
b_{k,\infty} = \lim_{n\to\infty} b_{k,n}, \quad \text{and} \quad
b_{k,\infty}-a_{k,\infty}\ge \eps, \quad  k=1,\ldots, L
\]
where the last inequalities come from \eqref{eq:epskn}.

\bigskip

The sequence $d_n(\cdot)$ is bounded (in $L^2$-norm):
\[
\|d_n(\cdot)\| \le \|\widehat\beta_{L^2}(\cdot)\| + \|d_n(\cdot) -
\widehat\beta_{L^2}(\cdot)\|
\le R + \sqrt{m+1}
\]
with \eqref{eq:mmmm}, where $R=\|\widehat\beta_{L^2}(\cdot)\|$. Moreover
\[
\|d_n(\cdot)\|^2 = \sum_{k=1}^L
\alpha^2_{k,n}\big(b_{k,n}-a_{k,n}\big) 
\ge \eps \sum_{k=1}^L
\alpha^2_{k,n}.
\]
Hence, each sequence $\{\alpha_{1,n}\}_n$, \ldots,
$\{\alpha_{L,n}\}_n$ is bounded. Thus, by further extracting
subsubsequences,
we can assume that, for $k=1,\ldots, L$,
\[
\lim_{n\to\infty}\alpha_{k,n}=\alpha_{k,\infty}
\]

\bigskip

Finally, by setting
\[
d_\infty(\cdot) = \sum_{k=1}^L \alpha_{k,\infty} \mathds 1\{t \in (a_{k,\infty},\ b_{k,\infty})\}
\]
we can easily prove that $d_n(\cdot)$ tends to $d_\infty(\cdot)$ in
$L^2$-norm and that $d_\infty(\cdot)\in \mathcal E_{K_0}^{\eps}$. And,
with (\ref{eq:mmmm})
\[
m = \| d_\infty(\cdot)-\widehat\beta_{L^2}(\cdot) \|
\]
which concludes the proof. \hfill \qed
\subsection{Topological properties of $\mathcal E_K$}
\label{sec:proof3}
\begin{proposition} \label{pro:topo} Let $K\ge 1$. 
  \begin{itemize}
  \item[\it (i)] The convex hull of $\mathcal E_K$ is $\mathcal E$. 
  \item[\it (ii)]  Under the $L^2(\mathcal T)$-topology, the closure of $\mathcal E$ is $L^2(\mathcal T)$.
  \end{itemize}
\end{proposition}
\begin{proof}
The result of \textit{(ii)} is rather classical, see, e.g., \citet{rudin1986real}.
The convex hull of $\mathcal E_K$ includes any step function. Indeed, any step function can be written as a convex combination of simple $a\mathbf{1}\{t\in I\}$'s which all belongs to $\mathcal E_K$. Moreover, $\mathcal E$ is convex because it is a linear space. Hence claim \textit{(i)} is proven. 
\end{proof}

For a given $K$, the set of functions $\mathcal{E}_K$ is not suitable to define a projection of $\hat{\beta}_{L^2}(\cdot)$. Indeed, let $\{d_n(\cdot)\}$ be a minimizing sequence of the set $\big\{ \| d(\cdot) - \hat{\beta}_{L^2}(\cdot) \|, \, d(\cdot) \in \mathcal{E}_K(\cdot) \big\}$, so 
\begin{equation*}
m= \inf \Big\{\| d(\cdot) - \widehat{\beta}_{L^2}(\cdot)\|,\ d(\cdot)
\in \mathcal E_{K} \Big\}
\le \| d_n(\cdot) - \widehat{\beta}_{L^2}(\cdot)\|  \le m +
2^{-n}.
\end{equation*}
Knowing that $\hat{\beta}_{L^2}(\cdot)$ and $d_n(\cdot)$ belong to $L^2$ for all $n$, we have 
$$d_n(.) \in \mathcal{E}_K \cap \mathcal{B}_{L^2}(R + m +1), \qquad \text{for all } n,$$
 where $\mathcal{B}_{L^2}(r)$ is the $L^2$-ball of radius $r$ around the origin. Note that $\mathcal{E}_K \cap \mathcal{B}_{L^2}(R + m +1)$ is not a compact set, for example consider $d_n(t) = \sqrt{n} \, \mathbf{1}\{ t \in [0,\frac{1}{n}]\}$. Hence it is not possible to extract a subsequence of $\{d_n(\cdot)\}$ which converges to a $d_{\infty}(\cdot) \in \mathcal{E}_K$ such that $\| d(\cdot) - \hat{\beta}_{L^2}(\cdot) \|=m$. 

\section{Details of the implementations}
\subsection{Gibbs algorithm and Full conditional distributions}
\label{app:conditional}
The full conditional distributions for the Gibbs Sampler in Section~\ref{sec:implementation} are the following,
    \begin{align*}
        \mu, \beta^* | y, \sigma^2 ,m,\ell& \sim \mathcal{N}_{K+1} \left( (\underbar{x}^T \underbar{x} + \underbar{V})^{-1} \underbar{x} y ~ , ~ \sigma^2 (\underbar{x}^T \underbar{x} + \underbar{V})^{-1} \right), \\
        \sigma^2 | y, \mu, \beta^*,m,\ell   & \sim \Gamma^{-1} \left(a + \frac{n+K+1}{2},b + \frac{1}{2}\text{SSE} + \frac{1}{2}\left\| \beta^* - \eta \right\|^2_{V^{-1}} \right), \\            
        \pi \left( m_k | y, \mu, \beta^*,\sigma^2,m_{-k},\ell \right) & \propto \exp \left( - {\text{SSE}}/{2\sigma^2} \right) 
      \times \pi(\beta^\ast| m, \ell, \sigma^2) \\
        \pi \left( \ell_k | y, \mu, \beta^*,\sigma^2,m,\ell_{-k} \right) & \propto \exp \left( - {\text{SSE} }/{2\sigma^2}  \right) \times \pi(\ell_k) \times \pi(\beta^\ast| m, \ell, \sigma^2)
    \end{align*}
where $\text{SSE}=\left\| y  -\mu\mathbf{1}_n - x_.(\mathcal{I}_.) \beta^* \right\|^2$, $ \underbar{x} =  \Big(\mathbf{1}_n ~ \mid ~ x.(\mathcal{I}.) \Big) $, and 
$$ \underbar{V} = \begin{pmatrix}
    v_0^{-1} & 0 \\
    0 & n^{-1} \Big(x_.(\mathcal{I}_.)^T x_.(\mathcal{I}_.) + v I_K \Big)
\end{pmatrix}. $$
The full conditional distributions for the hyperparameters $m_k$ and $\ell_k$ are unusual distributions. 
As the covariate curves $x_i$ are observed on a grid $\mathcal{T}_G = (t_j)_{j=1,\dots,p}$, we consider that $m_k$ belongs to $\mathcal{T}_G$ and $\ell_k$ is such that $m_k \pm \ell_k \in \mathcal{T}_G$. Thus, the number of possible values for $m_k$ and $\ell_k$ is finite and the full conditional distributions of $m_k$ and $\ell_k$ are easily computable.
\subsection{Simulated annealing algorithm}
\label{ann:SANN}
We give in this section the details of the Simulated Annealing algorithm we use. Let $\tilde{\Theta}_{K_0} = \bigotimes_{K=1}^{K_0} \big(K, \Theta_K \big)$ where $\Theta_K$ is the space of all $\theta=(\beta_1^*, \dots, \beta_K^*, m_1, \dots, m_K, \ell_1, \dots, \ell_K) $ and let the function $C(d(\cdot)) = \big\| d(\cdot) - \hat{\beta}_{L^2}(\cdot) \big\|^2$.  \\

	\hrule \vspace*{-.1cm} 
	\textbf{Algorithm :} \textit{Simulated Annealing} \vspace{.1cm}
	\hrule
	\begin{itemize}
		\itemb Initialize: a deterministic decreasing schedule of temperature $(\tau_i)_{i=1,\dots,N_{\text{SANN}}}$, a value of $K_0$ and an initial vector $(K_{(0)}, \theta_{(0)} ) \in \tilde{\Theta}_{K_0}$.
		\itemb Compute the function $\beta_{(0)}(t)$ from $(K_{(0)}, \theta_{(0)} )$.
		\itemb Repeat for $i$ from $1$ to $N_{\text{SANN}}$ : 
		\begin{itemize}
			\itemb Choose randomly a move from $(K_{(i-1)}, \theta_{(i-1)} )$ to $(K', \theta' )$ among : 
			\begin{enumerate}
				\item propose a new ${\beta_k^*}'$ for an arbitrary $k \leq K_{(i-1)}$,
				\item propose a new $m_k'$ for an arbitrary $k \leq K_{(i-1)}$,
				\item propose a new $\ell_k'$ for an arbitrary $k \leq K_{(i-1)}$,
				\item propose to append a new interval $({\beta^* }', m',\ell')$ or
				\item propose to drop out an interval $(\beta^*_k , m_k,\ell_k)$ for an arbitrary $k \leq K_{(i-1)}$.
			\end{enumerate}  
			\itemb Compute the function $\beta'(t)$ from the proposal $(K',\theta')$.
			\itemb Compute the acceptance ratio $$ \alpha  = \min \left\{ 1 , \exp \left( \frac{C(\beta'(\cdot)) - C\big(\beta_{(i)}(\cdot)\big)}{\tau_i} \right) \right\}. $$
			\itemb Draw $u$ from $\text{Unif}(0,1)$.
			\itemb If $u < \alpha$, $(K_{(i)}, \theta_{(i)} ) = (K', \theta' )$ (move accepted), \\
			~~~ ~~~ ~~~  else $(K_{(i)}, \theta_{(i)} ) = (K_{(i-1)}, \theta_{(i-1)} )$ (move rejected).
			\itemb Compute the function $\beta_{(i)}(t)$ from $(K_{(i)}, \theta_{(i)} )$.
		\end{itemize}
		\itemb Return the iteration $(K_{(i)}, \theta_{(i)} )$ minimizing the criteria $C(.)$. 
	\end{itemize}	
	\hrule \vspace{0.3cm}	
For the schedule of temperature, we use by default a logarithmic schedule (see \citealp{Belisle1992}), which is given for each iteration $i$ by 
\begin{equation}
	\text{Te} / \log \left((i-1) + e \right),
\end{equation}
where $\text{Te}$ is a parameter to calibrate and corresponds to the
initial temperature. The result of the Simulated Annealing algorithm
is sensitive to the scale of $\text{Te}$ and it is quite difficult to
find an a priori suitable value. For example, if the initial
temperature is too small, almost all the proposed moves are rejected
during the algorithm. On the opposite, if it is too large, they are
almost all accepted. So, we run the algorithm a few times and each time $\text{Te}$ is determined with respect to the previous runs. For instance, if for a run the moves are always rejected or always accepted, the initial temperature for the next run is accordingly adjusted. Only 2 or 3 runs are sufficient to find a suitable scale of $\text{Te}$. 



\end{document}